\documentclass{article}
\usepackage{ijcai23}

\usepackage{times}
\usepackage{soul}
\usepackage{url}
\usepackage[hidelinks]{hyperref}
\usepackage[utf8]{inputenc}
\usepackage[small]{caption}
\usepackage{graphicx}
\usepackage{amsmath}
\usepackage{booktabs}
\usepackage{tikz}
\usepackage{multirow}
\usepackage{algorithm}
\usepackage{algorithmic}
\usepackage[switch]{lineno}
\usepackage{xspace}

\usepackage{latexsym}
\usepackage{amsfonts}
\usepackage{amssymb}
\usepackage{enumitem}

\usepackage{amsthm}
\newtheorem{theorem}{Theorem}

\newtheorem{example}[theorem]{Example}

\newtheorem{lemma}[theorem]{Lemma}
\newtheorem{corollary}[theorem]{Corollary}

\usepackage{tikz}
\usetikzlibrary{automata}
\usetikzlibrary{positioning,arrows,fit,calc}
\usetikzlibrary{decorations.pathmorphing}
\usetikzlibrary{shapes}
\usetikzlibrary{backgrounds}
\usepackage{multirow}
\usepackage{booktabs}
\usepackage{xcolor,colortbl}
\usepackage{setspace}
\usetikzlibrary{decorations.pathreplacing}
\usetikzlibrary{matrix}

\usepackage{thmtools,thm-restate}

\usepackage{pifont}

\newcommand{\avec}[1]{\boldsymbol{#1}}

\newcommand{\LTL}{\textsl{LTL}}

\newcommand{\D}{\mathcal{D}}
\newcommand{\nxt}{{\ensuremath\raisebox{0.25ex}{\text{\scriptsize$\bigcirc$}}}}

\newcommand{\op}{{\boldsymbol{o}}}

\newcommand{\Abox}{\mathcal D}

\newcommand{\q}{\ensuremath{\boldsymbol{q}}\xspace}
\newcommand{\s}{\boldsymbol{q}}

\newcommand{\sig}{\mathsf{sig}}

\newcommand{\dep}{\mathit{tdp}}

\usepackage[mathscr]{euscript}
 \let\mathscr\relax
\usepackage[scr]{rsfso}

\tikzset{
  basic box/.style = {
    shape = rectangle,
    align = center,
    draw  = #1,
    rounded corners},
  header node/.style = {
    font          = \strut\Large\ttfamily,
    text depth    = +0pt,
    fill          = white,
    draw},
  header/.style = {%
    inner ysep = +1.5em,
    append after command = {
      \pgfextra{\let\TikZlastnode\tikzlastnode}
      node [header node] (header-\TikZlastnode) at (\TikZlastnode.north) {#1}
    }
  },
  hv/.style = {to path = {-|(\tikztotarget)\tikztonodes}},
  vh/.style = {to path = {|-(\tikztotarget)\tikztonodes}},
  fat blue line/.style = {ultra thick, blue}
}

\newcommand{\slin}[3]
{
{
\scriptsize
\draw[-, very thick] ($({#1}) +(0,0.09)$) -- node[above=0.1] {#2}  node[below=0.1] {#3} ($({#1}) +(0,-0.09)$);
}
}

\newcommand{\NL}{\textsc{NL}}
\newcommand{\NP}{\textsc{NP}}
\newcommand{\coNP}{\textsc{coNP}}

\newcommand{\PTime}{\textsc{P}}
\newcommand{\NExpTime}{\textsc{NExpTime}}

\newcommand{\Dmc}{\ensuremath{\mathcal{D}}\xspace}

\newcommand{\Q}{\mathcal{Q}}

\newcommand{\Qd}{\mathcal{Q}[\Diamond]}
\newcommand{\Qnd}{\mathcal{Q}[\nxt\Diamond]}
\newcommand{\Qpd}{\mathcal{Q}_p[\Diamond]}
\newcommand{\Qpnd}{\mathcal{Q}_p[\nxt\Diamond]}
\newcommand{\Qint}{\ensuremath{\mathcal{Q}_{\textit{in}}}\xspace}

\newcommand{\sQpnd}{\mathcal{Q}^\sigma_p[\nxt\Diamond]}
\newcommand{\sQd}{\mathcal{Q}^\sigma[\Diamond]}

\newcommand{\sQpd}{\mathcal{Q}^\sigma_p[\Diamond]}

\newcommand{\fr}{\mathcal{F}}

\newcommand{\sep}{\mathsf{s}}
\newcommand{\sepms}{\mathsf{mss}}
\newcommand{\sepmg}{\mathsf{mgs}}

\title{Extremal Separation Problems for Temporal Instance Queries}

 \author{
 Jean Christoph Jung$^1$\and
 Vladislav Ryzhikov$^2$\and
 Frank Wolter$^3$\And
 Michael Zakharyaschev$^{2}$
 \affiliations
 $^1$Department of Computer Science, TU Dortmund University, Germany\\
 $^2$School of Computing and Mathematical Sciences, Birkbeck, University of London, UK\\
 $^3$Department of Computer Science, University of Liverpool, UK\\
 \emails
 jean.jung@tu-dortmund.de, \{v.ryzhikov,m.zakharyaschev\}@bbk.ac.uk, wolter@liverpool.ac.uk
 }

\begin{document}


\maketitle


\begin{abstract}
%
The separation problem for a class $\mathcal{Q}$ of database queries is to find a query in $\mathcal{Q}$ that
distinguishes between a given set of `positive' and `negative' data examples. Separation provides explanations of examples and underpins the query-by-example paradigm
to support database users in constructing and refining queries.
As the space of all separating queries can be large, it is helpful to succinctly represent this space
by means of its most specific (logically strongest) and general (weakest) members. We investigate this
extremal separation problem for classes of instance queries formulated in linear temporal logic \LTL{} with the
operators conjunction, `next', and `eventually'\!. Our results range from tight
complexity bounds for verifying and counting extremal separators to algorithms computing them.
%
\end{abstract}


\section{Introduction}\label{intro}

The separation (aka fitting or consistency) problem for a class $\mathcal{Q}$ of queries is to find some $\q\in \mathcal{Q}$ that separates a given set $E=(E^{+},E^{-})$ of positive and negative  data examples in the sense that $\D\models \q$ for all $\D\in E^{+}$, and $\D\not\models \q$ for all \mbox{$\D\in E^{-}$}. Separation underpins the query-by-example approach, which aims to support database users in constructing queries and schema mappings with the help of data examples~\cite{DBLP:journals/tods/AlexeCKT11,martins2019reverse}, inductive logic programming~\cite{DBLP:journals/ml/CropperDEM22}, and, more recently, automated feature extraction, where separating queries are proposed as features in classifier engineering~\cite{DBLP:journals/tods/KimelfeldR18,DBLP:journals/jcss/BarceloBDK21}. Separating queries (and, more generally, formulas) also underpin recent logic-based approaches aiming to explain positive and negative data examples given by applications~\cite{DBLP:conf/nesy/SarkerXDRH17,DBLP:conf/aips/CamachoM19,DBLP:conf/tacas/RahaRFN22}.

The space of all separating queries, denoted $\sep(E,\mathcal{Q})$, forms a convex subset of $\mathcal{Q}$ under the containment (or logical entailment) relation $\q \models \q'$ between queries $\q$, $\q'$, and is an instance of the more general  version spaces~\cite{DBLP:journals/ai/Mitchell82}. If finite, $\sep(E,\mathcal{Q})$ can be represented by its extremal elements: the most specific (logically strongest) and most general (logically weakest) separators in $\mathcal{Q}$. 
In fact, many known algorithms  check the existence of separators by
looking for a most specific
one~\cite{DBLP:conf/icdt/CateD15,DBLP:conf/icdt/Barcelo017,GuJuSa-IJCAI18,DBLP:conf/ijcai/FunkJLPW19}.
This is not surprising as query classes are often closed under
$\land$, and so the conjunction of all separators gives the unique
most specific one. Dually, the unique most general separator is given
by the disjunction of all separators if the query class is closed
under $\lor$, which is a less common assumption. For the case of first-order queries constructed using $\land$ and existential quantifiers (conjunctive queries or CQs), a systematic
study of extremal separation has recently been conducted in the award winning~\cite{DBLP:conf/pods/CateDFL23}. 

Here, we study extremal separation for temporal instance
queries. Data instances take the form $(\delta_{0},\dots,\delta_{n})$
describing temporal evolutions, where the $\delta_{i}$ are the sets of
atomic propositions that are true at time $i$. Queries are formulated
in the fragment of linear temporal logic \LTL{} with the operators
$\land$, $\nxt$ (next), and $\Diamond$ (eventually). These \emph{$\nxt\Diamond$-queries} and its subclass of \emph{$\Diamond$-queries} without $\nxt$, are obtained by restricting CQs to propositional temporal data and form the core of most temporal query languages proposed in the database and knowledge representation literature~\cite{DBLP:reference/db/ChomickiT18a,BaBL-JWS15,BoLT-JWS15,DBLP:journals/ai/ArtaleKKRWZ21}. We are particularly interested in subclasses
of the classes of $\nxt\Diamond$-and $\Diamond$-queries that are not closed under $\land$ and take the form of path queries:
\begin{equation}\label{dnpath}
	\q = \rho_0 \land \op_1 (\rho_1 \land \op_2 (\rho_2 \land \dots \land \op_n \rho_n) )
\end{equation}
with $\op_i \in \{\nxt, \Diamond\}$ and conjunctions $\rho_i$ of atoms.
The temporal patterns expressed by path queries correspond to common subsequences, subwords, and combinations thereof, which have been investigated
in the string pattern matching literature for more than 50 years. Their
applications range from computational linguistics to bioinformatics and revision control systems~\cite{878178,4609376,DBLP:conf/focs/AbboudBW15,DBLP:journals/cor/BlumDSJLMR21}. In fact, our results can also be interpreted and applied in that research tradition and our techniques combine both logic and automata-based methods with pattern matching. We note that not admitting
$\lor$ in our query languages is crucial for finding this type of patterns and adding it would often trivialise separation. 
\begin{example}\label{exm:1}\em
Suppose the first two sequences of events shown below are `positive'\!, the third one is `negative'\!, and our task is to explain this phenomenon using path queries. The space of \\
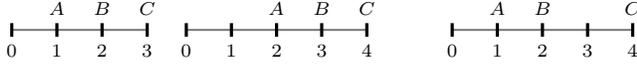
\centerline{
\begin{tikzpicture}[nd/.style={draw,thick,circle,inner sep=0pt,minimum size=1.5mm,fill=white},xscale=0.6]
\draw[thick,gray,-] (0,0) -- (3,0);
\slin{0,0}{}{\scriptsize$0$};
\slin{1,0}{\scriptsize$A$}{\scriptsize$1$};
\slin{2,0}{\scriptsize$B$}{\scriptsize$2$};
\slin{3,0}{\scriptsize$C$}{\scriptsize$3$};
\end{tikzpicture}
\
\begin{tikzpicture}[nd/.style={draw,thick,circle,inner sep=0pt,minimum size=1.5mm,fill=white},xscale=0.6]
\draw[thick,gray,-] (0,0) -- (4,0);
\slin{0,0}{}{\scriptsize$0$};
\slin{1,0}{}{\scriptsize$1$};
\slin{2,0}{\scriptsize$A$}{\scriptsize$2$};
\slin{3,0}{\scriptsize$B$}{\scriptsize$3$};
\slin{4,0}{\scriptsize$C$}{\scriptsize$4$};
\end{tikzpicture}
\qquad
\begin{tikzpicture}[nd/.style={draw,thick,circle,inner sep=0pt,minimum size=1.5mm,fill=white},xscale=0.6]
\draw[thick,gray,-] (0,0) -- (4,0);
\slin{0,0}{}{\scriptsize$0$};
\slin{1,0}{\scriptsize$A$}{\scriptsize$1$};
\slin{2,0}{\scriptsize$B$}{\scriptsize$2$};
\slin{3,0}{}{\scriptsize$3$};
\slin{4,0}{\scriptsize$C$}{\scriptsize$4$};
\end{tikzpicture}
}\\
possible explanations includes $\q_1 = \Diamond (A \land \nxt (B \land \nxt C))$, $\q_2 = \Diamond (A \land \nxt \nxt C)$ and $\q_3 = \Diamond (B \land \nxt C)$, all of which are true at $0$ in the positive examples and false in the negative one. In fact, $\q_1$ is the  unique most specific explanation, while $\q_2$ and $\q_3$ are the non-equivalent most general ones.
On the other hand, there exists no explanation in terms of $\Diamond$ only.
\end{example}
The (bounded-size) separator existence problem for various classes $\mathcal{Q}$ of
\LTL-queries has recently been studied
in~\cite{DBLP:conf/icgi/FijalkowL21,DBLP:conf/tacas/RahaRFN22,DBLP:conf/ijcai/FortinKRSWZ23}.
Our aim here is to investigate systematically the separator spaces
$\sep(E,\mathcal{Q})$ by determining the complexity of
verifying and counting most specific and general separators and giving algorithms computing them. On our way, we determine the complexity of entailment (aka containment in database theory) between queries and computing weakening and strengthening frontiers, which are the key to understanding $\sep(E,\mathcal{Q})$.
Intuitively, a weakening/strengthening frontier of $\q\in \Q$ is a set of queries properly weaker/stronger than $\q$ that form a boundary between $\q$ and all of its weakenings/strengthenings in $\Q$.

In detail, we first prove that query containment is in \PTime{} for
the class of $\Diamond$-queries and all of our classes of path queries.
Based on this result, we show that strengthening and weakening frontiers can be computed in polytime for path queries. This is also the case for $\Diamond$-queries and weakening frontiers but not for strengthening ones. It follows that checking whether a path query is a most specific/general separator and whether a $\Diamond$-query is a most general one are both in \PTime{}. In contrast, we establish \coNP-completeness of checking whether a $\Diamond$-query is a most specific separator and whether a path query is the unique most specific/general one. Using frontiers, we show for path queries that  the existence of unique most specific/general separators is in the complexity class US (for which unique SAT is complete)
and that counting the number of most general/specific separators is in $\sharp$\PTime. We show that these upper bounds are tight, sometimes  using the rich literature on algorithms for sequences (e.g., longest common subsequences). Our lower bounds mostly require only a bounded number of atomic propositions and an unbounded number of either  negative or positive examples.


These complexity results are complemented with algorithms for computing extremal separators in the majority of our query classes. The algorithms use a graph encoding of the input example set, associating (extremal) separators with certain paths in the graph. Complexity-wise, they are optimal, running in polytime if the number of examples is bounded and exponential time otherwise.




\subsection{Related Work}

Being inspired by the investigation of extremal separation for first-order conjunctive queries (CQs)~\cite{DBLP:conf/pods/CateDFL23}, our results turn out to be very different.
For instance, while separability by CQs is \NExpTime-complete and separating queries are exponential in the size of the examples (the extremal ones even larger), the extremal separation problems for \LTL-queries are often complete for SAT-related complexity classes, with separating queries being of polynomial size. For work on separation in the query-by-example paradigm we refer the reader to~\cite{DBLP:conf/sigmod/ZhangEPS13,DBLP:conf/pods/WeissC17,kalashnikov2018fastqre,deutch2019reverse,DBLP:conf/icdt/StaworkoW12,DBLP:conf/icdt/Barcelo017,DBLP:journals/tods/CohenW16,DBLP:conf/www/ArenasDK16} in the database context and to~\cite{GuJuSa-IJCAI18,DBLP:conf/gcai/Ortiz19,DBLP:conf/cikm/CimaCL21,DBLP:journals/ai/JungLPW22} in the context of KR.

Our contribution is also closely related to work on synthesising \LTL{}-formulas that explain the positive and negative data examples coming from an application~\cite{lemieux2015general,DBLP:conf/fmcad/NeiderG18,DBLP:conf/aips/CamachoM19,DBLP:conf/tacas/RahaRFN22,DBLP:conf/kr/FortinKRSWZ22,DBLP:conf/ijcai/FortinKRSWZ23}.
While concerned with separability of temporal data instances and, in particular, separability by \LTL-formulas of small size, the separator spaces $\sep(E,\mathcal{Q})$ themselves have not yet been investigated in this context.

\section{Data and Queries in \LTL}\label{prelims}

Fix some countably-infinite set of unary predicate symbols, called \emph{atoms}.
A \emph{signature}, $\sigma$,  is any finite set of atoms. A (\emph{temporal}) \emph{data instance} is any finite set  $\Abox \ne \emptyset$ of \emph{facts} $A(\ell)$ with an atom $A$ and a \emph{timestamp} $\ell \in \mathbb N$, saying that $A$ happened at $\ell$. The \emph{size} $|\Abox|$ of $\Abox$ is the number of symbols in it, with the timestamps given in \emph{unary}. Let $\max\Abox$ be the maximal timestamp in $\Abox$. Where convenient, we also write $\Abox$ as the word $\delta_0 \dots \delta_{\max\Abox}$ with $\delta_i = \{A  \mid A(i) \in \Abox\}$. The \emph{signature} $\sig(\Abox)$ of $\Abox$ is the set of atoms occurring in it.

We query data instances by means of \LTL-formulas, called
\emph{queries}, that are built from atoms (treated as propositional
variables) and the logical constant $\top$ (truth) using $\land$ and the temporal operators $\nxt$ (next time) and $\Diamond$ (sometime in the future). We consider the following classes of queries:
\begin{description}\itemsep=0pt
\item[$\Qnd$:] all $\nxt\Diamond$-\emph{queries};

\item[$\Qd$:] all $\Diamond$-\emph{queries} (not containing $\nxt$);

\item[$\Qpnd$:] \emph{path $\nxt\Diamond$-queries} of the form~\eqref{dnpath}, where the conjunctions of atoms $\rho_i$ are often treated as sets and the empty conjunction as $\top$;

\item[$\Qpd$:] all \emph{path $\Diamond$-queries} (not containing $\nxt$);

\item[$\Qint$:] \emph{interval-queries} of the form~\eqref{dnpath} with $\rho_{0}=\top$, $\rho_1 \ne \top$, $\op_{1}=\Diamond$, and $\op_{i}=\nxt$, for $i>1$.
\end{description}
Queries in $\Qint$ single out an interval of a fixed length starting
at some time-point $\geq 1$; $\q_{1}$--$\q_{3}$ from
Example~\ref{exm:1} are in $\Qint$.
$\Q^\sigma$ is the restriction of a class $\Q$ to a signature $\sigma$.
The \emph{temporal depth} $\dep(\q)$ of $\q$ is the maximum number of nested temporal operators in $\q$. The \emph{signature} $\sig(\q)$ of $\q$ is the set of atoms in $\q$; the \emph{size} $|\q|$ of $\q$ is the number of symbols in it.

The \emph{truth-relation} $\D,n \models \q$---saying that $\q$ is true in $\D$ at moment $n \in \mathbb N$---is defined as usual in temporal logic under the \emph{strict semantics}: $\D,n \models \top$ for all $n \in \mathbb N$; $\D,n \models A_i$ iff $A_i(n) \in \D$; $\D,n \models \nxt\q'$ iff $\D,n+1 \models \q'$; and $\D,n \models \Diamond \q'$ iff $\D,m \models \q'$, for some $m > n$.
A data instance $\D$ is called a \emph{positive example} for a query $\q$ if $\D,0 \models \q$; otherwise, $\D$ is a \emph{negative example} for $\q$.
Checking whether $\D$ is a positive (negative) example for our queries
$\q$ can obviously be done in polytime in $|\D|$ and
$|\q|$.

We write $\q \models \q'$ if $\D,0 \models \q$ implies $\D,0 \models \q'$ for all instances $\D$, and $\q \equiv \q'$ if $\q \models \q'$ and $\q' \models \q$, in which case $\q$ and $\q'$ are \emph{equivalent}. For example, for any query $\q$, we have $\Diamond\nxt\q \equiv \nxt\Diamond \q \equiv \Diamond\Diamond \q$. It follows that every path query in $\Qpnd$ is equivalent to a query of the form~\eqref{dnpath}, in which $\rho_n \ne \top$ and whenever $\rho_i = \top$, $0 < i < n$, then $\op_i = \op_{i+1}$; in this case we say that $\q$ is in \emph{normal form}. Unless indicated otherwise, we assume all path queries to be in normal form.


\smallskip
\noindent
\textbf{Sequences.} There is a close link between evaluating path queries and algorithms for finding patterns in strings \cite{DBLP:books/daglib/0020103}. A \emph{sequence} is a data instance $\D=\delta_{0}\dots\delta_{n}$  with $\delta_{0}=\emptyset$ and $|\delta_{i}| =1$, for $i >0$; a \emph{sequence query} is a path query of the form \eqref{dnpath} with $|\rho_{0}|= 0$ and $|\rho_{i}| = 1$, for $i>0$. Querying sequences using sequence queries corresponds to the following matching problems:
\begin{itemize}\itemsep=0pt
\item[--] for any sequence query $\q \in \Qpd$ of the form \eqref{dnpath}, we have $\D,0\models \q$ iff $\rho_{1}\dots\rho_{m}$ is a \emph{subsequence} of $\D$;
		
\item[--] for any sequence query $\q \in \Qint$ of the form~\eqref{dnpath}, we have $\D,0\models \q$ iff $\rho_{1}\dots\rho_{m}$ is a \emph{subword} of $\D$.
\end{itemize}


\section{Query Containment}\label{sec:cont}

The \emph{query containment problem} for a class $\Q$ of queries is to
decide whether $\q \models \q'$, for any given $\q,\q' \in \Q$. In
contrast to conjunctive queries, where query containment is
\NP-complete~\cite{Chandra&Merlin77}, query
containment is tractable for the majority of query classes defined
above:


\begin{theorem}\label{thm:containment}
The query containment problems for $\Qpnd$, $\Qd$ \textup{(}and their subclasses\textup{)} are all in \PTime{}.
\end{theorem}

To prove Theorem~\ref{thm:containment}, suppose first that we are given two queries $\q,\q' \in \Qpnd$, where $\q$ takes of the form~\eqref{dnpath} and $\q'= \rho_{0}' \wedge \op_{1}' (\rho_{1}' \wedge \cdots \wedge\op_{m}' \rho_{m}')$. Denote by $[m]$ the closed interval $[0,m] \subseteq \mathbb N$. A function $h \colon [m] \rightarrow [n]$ is \emph{monotone} if $h(i) < h(j)$ whenever $i <j$.
Then tractability of containment for path queries follows from the criterion below, which is proved in the appendix by induction on $\dep(\q')$:

\begin{lemma}\label{lem:containment}
Let $\q,\q' \in \Qpnd$. Then $\q \models \q'$ iff there is a monotone function $h \colon [m] \rightarrow [n]$ such that $h(0)=0$ and, for all $i \in [m]$, we have $\rho_{i}' \subseteq \rho_{h(i)}$ and if $\op_{i+1}'=\nxt$, then $\op_{h(i+1)}=\nxt$ and $h(i+1)=h(i)+1$.
\end{lemma}

We refer to any function $h$ defined in Lemma~\ref{lem:containment} as a \emph{containment witness} for the pair $\q, \q' \in \Qpnd$.

Suppose now $\q \in \Qd$. As shown in~\cite{DBLP:conf/kr/FortinKRSWZ22}, we can convert $\q$ in polytime to an equivalent query in the \emph{normal form} $\rho \wedge \q_{1} \wedge \cdots \wedge \q_{n}$, where $\rho$ is a conjunction of atoms and each $\q_{i}$ is in $\Qpd$ and starts with $\Diamond$. Tractability of containment for $\Qd$-queries follows from:

\begin{lemma}\label{lem:decomp}
If $\q=\rho \wedge \q_{1} \wedge \cdots \wedge \q_{n}\in \Qd$ is in normal form, $\q'\in \Qpd$ and $\q\models \q'$, then there is $\q_i$, $1 \le i \le n$, with $\rho \wedge \q_{i} \models \q'$.
\end{lemma}
\begin{proof}
In the detailed proof given in the appendix, we show that, assuming $\rho \wedge \q_{i} \not \models \q'$ for all $i$, we can convert the $\q_i$ into a data instance $\D$ with $\D,0 \models \q$ and $\D,0 \not\models \q'$.
\end{proof}

For queries $\q \in \Qnd$, Lemma~\ref{lem:decomp} does not hold:

\begin{example}\em
Let $\q = \q_1 \land \q_2 \land \q_3 \land \q_4$, where
\begin{align*}
& \q_{1}= \Diamond (a \wedge \nxt ((a \wedge b) \wedge \nxt a)),\\
& \q_{2} =\Diamond (b \wedge \nxt ((a \wedge b) \wedge \nxt b)),\\
& \q_{3}=\Diamond (a \wedge \nxt ((a \wedge b) \wedge \nxt b)),\\
& \q_4 =  \Diamond (a \land \Diamond (b \land \Diamond (a \land \Diamond (b \land \Diamond (a \land \Diamond b))))),
\end{align*}
and $\q' = \Diamond (b \wedge \Diamond ((a \wedge b) \wedge \Diamond a))$. Then $\q\models \q'$ but $\q_i \not\models \q'$, for any $i$, $1 \le i \le 4$.
\end{example}

At the moment, the question whether containment of $\Qnd$-queries is tractable remains open.


\section{Example Sets and Separating Queries}
\label{sec:example}

An \emph{example set} $E$ is a pair $(E^+, E^-)$ of finite sets $E^+ \ne \emptyset$ and $E^-$ of data instances. We say that a query $\q$ \emph{separates} $E$ if all $\D \in E^+$ are positive examples and all $\D\in E^-$ are negative examples for $\q$. Denote by $\sep(E,\mathcal{Q})$ the set of queries in a class $\mathcal{Q}$ separating $E$ and call $E$ $\mathcal{Q}$-\emph{separable} if $\sep(E,\mathcal{Q}) \ne \emptyset$. Our general aim is to understand the structure of $\sep(E,\mathcal{Q})$ for various important query classes $\Q$.
%

We consider queries in $\Q$ \emph{modulo equivalence}, not distinguishing between $\q \equiv \q'$. In this case, the relation $\models$ is a partial order on $\Q$.
For any $\q \in \sep(E,\mathcal{Q})$, we clearly have
\begin{itemize}
\item[--] $\dep(\q) \le \min \{\max \D \mid \D \in E^+\}$,

\item[--] $\sig(\q) \subseteq \bigcap \{ \sig (\D) \mid \D \in E^+\}$,
\end{itemize}
so $\sep(E,\mathcal{Q})$ is finite. We refer to the $\models$-minimal elements of $\sep(E,\mathcal{Q}) \ne \emptyset$ as \emph{most specific $\Q$-separators of} $E$ and to the $\models$-maximal elements as \emph{most general $\Q$-separators of} $E$ (modulo $\equiv$); they comprise the sets $\sepms(E,\mathcal{Q})$ and $\sepmg(E,\mathcal{Q})$, respectively. If these sets are singletons, we call their only element the \emph{unique most specific} and, respectively, \emph{unique most general $\Q$-separator of} $E$. Note that the former always exists if $\sep(E,\mathcal{Q}) \ne \emptyset$ and $\Q$ is closed under $\land$.

\begin{example}\em
$(i)$ Suppose $\D^+ = \{A(0),B(0),A(1),B(1)\}$, $\D^- = \{A(0)\}$ and $E = (\{\D^+\}, \{\D^-\})$. The separator space $\sep(E,\Qpd)$ is shown below as a Hasse diagram with arrows indicating the partial order $\models$ (and $\top,A \notin \sep(E,\Qpd)$), so\\
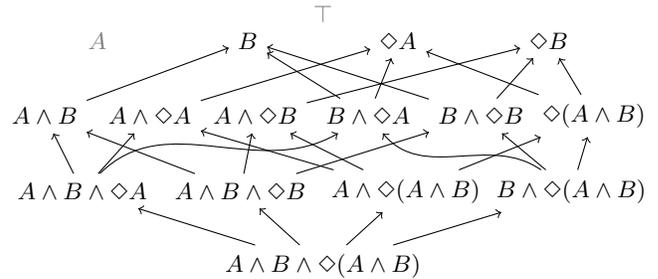
\centerline{
\begin{tikzpicture}[nd/.style={draw,thick,circle,inner sep=0pt,minimum size=1.5mm,fill=white}]

\node at (0,0) (n0) {\small $A \land B \land \Diamond (A \land B)$};

\node at (-3.2,1) (n1) {\small$A \land B \land \Diamond A$} edge[<-] (n0);

\node at (-1.1,1) (n2) {\small$A \land B \land \Diamond B$} edge[<-] (n0);

\node at (1.1,1) (n3) {\small$A \land \Diamond (A \land B)$} edge[<-] (n0);

\node at (3.3,1) (n4) {\small$B \land \Diamond (A \land B)$} edge[<-] (n0);

\node at (-3.7,2) (n21) {\small$A \land B$} edge[<-] (n1) edge[<-] (n2);

\node at (-2.3,2) (n22) {\small$A \land \Diamond A$} edge[<-] (n1) edge[<-] (n3);

\node at (-.9,2) (n23) {\small$A \land \Diamond B$} edge[<-] (n2) edge[<-] (n3);

\node at (.6,2) (n24) {\small$B \land \Diamond A$} edge[<-, out=-150,in=45] (n1) edge[<-, out=-50, in = 145] (n4);

\node at (2.1,2) (n25) {\small$B \land \Diamond B$} edge[<-] (n2) edge[<-] (n4);

\node at (3.6,2) (n26) {\small$\Diamond (A \land B)$} edge[<-] (n3) edge[<-] (n4);

\node at (-3, 3) (n31) {\small\textcolor{gray}{$A$}};

\node at (-1,3) (n32) {\small$B$} edge[<-] (n24) edge[<-] (n25) edge[<-] (n21);

\node at (1,3) (n33) {\small$\Diamond A$} edge[<-] (n22) edge[<-] (n24) edge[<-] (n26);

\node at (3,3) (n34) {\small$\Diamond B$} edge[<-] (n23) edge[<-] (n25) edge[<-] (n26);

\node at (0,3.35) (ntop) {\small \textcolor{gray}{$\top$}};

\end{tikzpicture}
}
\\
the most general $\Qpd$-separators of $E$ comprise the set
$\sepmg(E,\Qpd)= \{B, \Diamond A, \Diamond B\}$ and the unique most specific  one is $A\land B \land \Diamond(A\land B)$.

$(ii)$ For the example set $E = (\{\Abox^+_1,\Abox^+_2 \}, \{\Abox^-_1\})$ with
\begin{equation*}
\Abox^+_1 = \{A(1), B(2) \},  \Abox^+_2 = \{B(1), A(2)\}, \Abox^-_1 = \{C(0)\},
\end{equation*}
$\sep(E,\Qpd) = \{\Diamond A, \Diamond B\} = \sepms/\sepmg(E,\Qpd)$ but there is no unique most specific/general separator. In contrast, $\sep(E,\Qd) = \{\Diamond A, \Diamond B, \Diamond A \land \Diamond B\}$ has the unique most specific separator $\Diamond A \land \Diamond B$ but no unique most general one.

$(iii)$ One can show that $\sepms/\sepmg(E,\Q)$ always contains a  \emph{longest/shortest separator} of $E$ in $\Q$ (of largest/smallest temporal depth). To illustrate, let $E = (\{\Abox^+_1\}, \{\Abox^-_1, \Abox^-_2, \Abox^-_3\})$,
\begin{align*}
&\Abox^+_1 = \{B(0), C(0), A(1) \},\\
& \hspace*{1cm}\Abox^-_1 = \{B(0)\},\ \Abox^-_2 = \{C(0)\},\ \Abox^-_3 = \{A(1)\}.
\end{align*}
Then $\sepmg(E,\Qpd) = \{B \land C, B \land \Diamond A,C \land \Diamond A\}$, where $B \land C$ of depth 0 is the shortest separator of $E$ in $\Qpd$.
\end{example}

Our main concern is the following three algorithmic problems for query classes $\Q \subseteq \Qnd$ with input $E$ and $\q\in \mathcal{Q}$:
\begin{description}\itemsep=0pt
\item[(most specific/general separator verification):] decide whether $\q$ is an element of $\sepms(E,\mathcal{Q})$\,/\,$\sepmg(E,\mathcal{Q})$;

\item[(counting most specific/general separators):] count the elements of $\sepms(E,\mathcal{Q})$\,/\,$\sepmg(E,\mathcal{Q})$;

\item[(computing a most specific/general separator):] construct some query in $\sepms(E,\mathcal{Q})$\,/\,$\sepmg(E,\mathcal{Q})$.
\end{description}
We are particularly interested in deciding whether there is a unique most specific/general separator and computing it.
To achieve our aims, we obviously should be able to decide whether $\q \in \sep(E,\mathcal{Q})$ (\emph{separator verification}) and whether $\sep(E,\mathcal{Q})\ne \emptyset$ (\emph{separator existence}).  As mentioned in Section~\ref{prelims}, separator verification is in \PTime.
We are also interested in the case when the number of positive or negative examples in $E$ is bounded.
%
%
%
%
%
%
%
%
The table below summarises the complexities of separator existence for our query classes, where $\mathsf{b}^+$\,/\,$\mathsf{b}^-$\\
\centerline{
\begin{tabular}{c|c|c|c}
		\hline
		{\small separator existence} & $\mathsf{b}^+\!\!,\, \mathsf{b}^-$ & $\mathsf{b}^+$ & $\mathsf{b}^-$ {\small or unbounded} \\
\hline
$\Qpd\,/\,\Qpnd$ & {\small in} $\PTime$ & $\NP${\small-c.} & $\NP${\small-c.}    \\[2pt]
$\Qd\,/\,\Qnd\,/\,\Qint$ & {\small in} $\PTime$ & {\small in} $\PTime$ & $\NP${\small-c.} \\\hline
\end{tabular}%
}\\
means that $|E^+|$\,/\,$|E^-|$ is bounded, and one can assume a bounded signature. Except for $\Qint$, these results are shown in \cite{DBLP:conf/ijcai/FortinKRSWZ23}. The proofs use techniques developed for the \emph{longest common subsequence problem} (given $k>0$ and a set $E^{+}$ of sequences, is there a sequence query in $\Qpd$ of depth $\geq k$ that is a subsequence of all $\D\in  E^{+}$; see Section~\ref{prelims})
and separator existence for sequence queries in $\Qpd$ \cite{DBLP:journals/jacm/Maier78,DBLP:journals/tcs/Fraser96}. The $\NP$-upper bound for $\Qint$ is trivial; the lower one follows from the proof of Theorem~\ref{lem:consnew}, and tractability for $\mathsf{b}^+$ is ensured by the observation that there are polynomially-many relevant intervals in the $E^+$-examples.

This $\NP$-lower bound shows that great care is needed when transferring techniques from the literature on algorithms for sequences to our framework as separability of example sets of sequences using sequence queries in $\Qint$ is easily seen to be in $\PTime$ even for unbounded example sets and signatures.

A key to the extremal separator problems above is the following  notions of strengthening and weakening frontiers.


\section{Strengthening and Weakening Frontiers}
\label{sec:frontiers}

Let $\Q$ be a class of queries and $\q \in \Q$.
A set $\fr \subseteq \mathcal{Q}$ is called a \emph{strengthening frontier} for $\q$ in $\Q$ if
\begin{itemize}\itemsep=0pt
\item[--] for any $\q'\in \fr$, we have $\q'\models \q$ and $\q'\not\equiv \q$;

\item[--] for any $\q'' \in \mathcal{Q}$, if $\q''\models \q$ and $\q''\not\equiv \q$, then there is $\q'\in \fr$ such that $\q''\models \q'$.
\end{itemize}
%
%
A set $\fr \subseteq \mathcal{Q}$ is called a \emph{weakening frontier} for $\q$ in $\Q$ if
\begin{itemize}\itemsep=0pt
\item[--] for any $\q'\in \fr$, we have $\q\models \q'$ and $\q'\not\equiv \q$;

\item[--] for any $\q'' \in \mathcal{Q}$, if $\q\models \q''$ and $\q''\not\equiv \q$, then there is $\q'\in \fr$ with $\q'\models \q''$.
\end{itemize}
%
Trivial strengthening/weakening frontiers for $\q$ comprise all queries that are properly stronger/weaker than $\q$ in $\Q$; our concern, however, is finding small frontiers.
We show now that, for $\sQpd$ and $\sQpnd$, one can compute a strengthening/weakening frontier for any given $\q$ in polytime.


%


\begin{theorem}\label{thm:firstfrontier}
  Let $\q \in \sQpnd$ be in normal form~\eqref{dnpath}. A
  strengthening frontier for $\q$ in $\sQpnd$ can be computed in
  polytime
  by applying once to $\q$ one of the following operations,
  for $i \in [n]$ and $\q_i = \rho_i \land \op_{i+1}(\rho_{i+1} \land
  \dots \land \op_n \rho_n)$\textup{:}
\begin{enumerate}\itemsep=0pt
\item extend some $\rho_{i}$ in $\q$ by some $A\in\sigma\setminus
  \rho_{i}$\textup{;}

\item replace some $\Diamond \q_i$ in $\q$ by $\Diamond (\top \land \Diamond \q_i)$\textup{;}

\item replace some $\op_{i}=\Diamond$ in $\q$ by $\nxt$ provided that the resulting query is in normal form\textup{;}

\item add $\op_{n+1}\rho_{n+1}$ at the end of $\q$, where $\op_{n+1} = \Diamond$ and $\rho_{n+1} = A$, for some $A\in \sigma$.
\end{enumerate}
If $\q \in \sQpd$, a strengthening frontier for $\q$ in $\sQpd$ can be
computed in polytime using operations 1, 2, and 4.
\end{theorem}
\begin{proof}
Let $\fr$ be the set of queries obtained by a single application of one of these operations to $\q$. By Lemma~\ref{lem:containment}, $\q'\models \q$ and $\q\not\equiv \q'$ for all $\q'\in \fr$. Let $\q'\models \q$ and $\q\not\equiv \q'$, for some $\q'\in \sQpnd$ of the form
$\q'= \rho_{0}' \wedge \op_{1}' (\rho_{1}' \wedge \cdots \wedge\op_{m}' (\rho_{m}'))$. Take a containment witness $h \colon [n] \rightarrow [m]$ for $\q',\q$. If $h$ is surjective, then $n=m$ and $h(i)=i$ for all $i\in [n]$, and so $\rho_{i}\subseteq \rho_{i}'$. As $\q\not\equiv \q'$, either $\op_{i}'=\nxt$ and $\op_{i}=\Diamond$, for some $i\in [n]$, or $\rho_{i}\subsetneq \rho_{i}'$, for some $i\in [n]$. In the former case, operation 3 gives $\q''\in \fr$ with $\q'\models \q''$; in the latter one, operation 1 gives such a $\q''$.
		
Suppose $h$ is not surjective. If there is $i<n$ such that $h(i+1)-h(i)\geq 2$, then $\q'\models \q''$ for a $\q''\in \fr$ given by operation 2. Otherwise $m$ is not in the range of $h$, and we get such a $\q''\in \fr$  by operation 4.
\end{proof}

\begin{example}\label{ex:lfq}\em
For $\sigma = \{A, B\}$ and $\q = \Diamond (A \land \nxt B)$,  operations 1--4 give the following strengthening frontier for $\q$:
\begin{align*}
& \Diamond (A \land B \land \nxt B),  \Diamond (A \land \nxt (A \land B)), \Diamond (\top \land \Diamond (A \land \nxt B)),\\
 & \nxt (A \land \nxt B),  
 \Diamond (A \land \nxt (B \land \Diamond A)),  \Diamond (A \land \nxt (B \land \Diamond B)).
\end{align*}
\end{example}

A weakening frontier can be constructed by reversing operations 1--3 from Theorem~\ref{thm:firstfrontier} and using a similar argument:

\begin{theorem}\label{thm:secondfrontier}
Let $\q \in \sQpnd$ be in normal form~\eqref{dnpath}. A weakening frontier for $\q$ in $\sQpnd$ can be computed in polytime by applying once to $\q$ one of the following operations, for $i \in [n]$ and $\q_i = \rho_i \land \op_{i+1}(\rho_{i+1} \land \dots \land \op_n \rho_n)$\textup{:}
\begin{enumerate}\itemsep=0pt
\item drop some atom from $\rho_{i}$\textup{;}

\item replace some $\Diamond(\top \wedge \Diamond \q_{i})$ in $\q$ by $\Diamond \q_{i}$\textup{;}

\item replace some $\op_{i}=\nxt$ by $\Diamond$.
\end{enumerate}
If $\q \in \sQpd$, a weakening frontier for $\q$ in $\sQpd$ can be computed in polytime using operations 1 and 2.
\end{theorem}

\begin{example}\label{ex:ufq}\em
For $\sigma= \{A, B\}$ and $\q = \Diamond (A \land \nxt B)$, operations 1-3 give the following weakening frontier for $\q$:
\begin{align*}
\Diamond (\top \land \Diamond B), \quad \Diamond A, \quad \Diamond(A \land \Diamond B).
\end{align*}
Note that the computed weakening frontier can be made smaller by omitting $\Diamond A$, which is weaker than $\Diamond(A \land \Diamond B)$.
\end{example}
We next consider frontiers for queries in $\sQd$. We say that $\q = \rho \wedge \q_{1}\wedge \cdots \wedge \q_{n}\in \sQd$ in normal form is \emph{redundancy-free} if it does not contain $\q_{i},\q_{j}$ with $i\ne j$ and $\q_{i}\models \q_{j}$. Clearly, for any $\q \in \sQd$, we can compute an equivalent redundancy-free $\q'\in \sQd$ in polytime.

\begin{theorem}\label{thm:thirdfrontier}
Let $\q=\rho \wedge \q_{1}\wedge \cdots \wedge \q_{n}\in \sQd$ be  redundancy-free. A weakening frontier for $\q$ in $\sQd$ can be computed in polytime by a single application of one of the following operations to $\q$\textup{:}
\begin{enumerate}\itemsep=0pt
\item drop some atom from $\rho$\textup{;}

\item replace some $\q_{i}$ by $\bigwedge \mathcal{F}_{i}$, where $\mathcal{F}_{i}$ is the weakening frontier for $\q_i$ in $\sQpd$ provided by Theorem~\ref{thm:secondfrontier}.
\end{enumerate}
\end{theorem}
\begin{proof}
Let $\mathcal{F}$ be the set of queries defined above. Clearly, $\q\models \q'$; as $\q$ is redundancy-free and in view of  Lemma~\ref{lem:decomp}, $\q\not\equiv \q'$ for all $\q'\in \mathcal{F}$. Suppose $\q\models \q'$ and $\q\not\equiv \q'$,  for some $\q'= \rho' \wedge \q_{1}' \wedge \cdots \land \q_{m}'$ in $\sQd$.
%
By Lemma~\ref{lem:decomp}, $\rho'\subseteq \rho$ and, for each $j$, $1 \le j \le m$, there exists $f(j)$, $1 \le f(j) \le n$, with $\q_{f(j)} \models \q_{j}'$.
If $\rho'\subsetneq \rho$, then operation 1 gives a $\q''\in \mathcal{F}$ with $\q''\models \q'$. Otherwise, as $\q \not\equiv \q'$, there is $\q_{i}$ such that $\q_{j}' \not\models \q_{i}$ for all $\q_{j}'$, $1 \le j \le m$.  Let $\q''$ be the query obtained from $\q$ by replacing $\q_{i}$ with $\bigwedge\mathcal{F}_{i}$ by operation 2. To establish $\q'' \models \q'$, it suffices to show that $f(j)=i$ implies $\bigwedge\mathcal{F}_{i}\models \q_{j}$. So suppose $f(j)=i$. Then $\q_{j}'\not\models \q_{i}$, and so, by the definition of $\mathcal{F}_{i}$, there is $\q_{i}''\in \mathcal{F}_{i}$ with $\q_{i}''\models \q_{j}$.
\end{proof}

However, strengthening frontiers for queries in $\sQd$ are not necessarily of polynomial size as shown by the following:

\begin{example}\label{expDiam}\em
We represent queries in $\sQd$ of the form~\eqref{dnpath} as $\rho_{0}\dots\rho_{n}$. For $\sigma=\{A_{1},A_{2},B_{1}, B_{2}\}$, 
let $\q_{1}= \emptyset (\s\sigma)^{n} \s$, $\q_{2} = \emptyset\sigma^{2n+1}$, and $\s=\{A_{1},A_{2}\}\{B_{1},B_{2}\}$. Using~\cite[Example 18]{DBLP:conf/kr/FortinKRSWZ22}, one can show that any strengthening frontier for the query $\q_1 \land \q_2$ in $\sQd$ is of size $O(2^n)$.
\end{example} 	

Note also that weakening frontiers for $\q \in \Qint^{\sigma}$ can be computed in polytime using operation 1 in Theorem~\ref{thm:secondfrontier} (if $i=1$ and $|\rho_{i}|=1$, we drop $\top$ and take $\Diamond (\rho_2 \land \nxt(\rho_{3} \land \dots \land \nxt \rho_n)$).
On the other hand, strengthening frontiers can be infinite:

\begin{example}\em
All $\Diamond (A \wedge \nxt^{n}A)$, $n > 0$, are in any strengthening frontier for $\q=\Diamond A$ in $\Qint^{\sigma}$, where $\sigma=\{A\}$.
\end{example}

As shown by Theorem~\ref{thm:mostspeciverif}, the lack of polytime computable strengthening frontiers in $\Qd$ affects the complexity of verifying most specific separators. In contrast, the lack of polytime computable strengthening frontiers in $\Qint$ turns out to be harmless. For $n\geq \dep(\q)$, call $\fr \subseteq \Q$ an \emph{$n$-bounded weakening/strengthen\-ing frontier} for $\q$ in $\Q$ if $\fr$ is a weakening/strengthen\-ing frontier for $\q$ in the class $\{\q'\in \Q \mid \dep(\q')\leq n\}$ (we assume $n$ to be given in unary).

\begin{theorem}\label{thm:ninterval}
 Weakening frontiers and $n$-bounded strengthening frontiers in $\Qint^{\sigma}$ can be computed in polytime.
\end{theorem}

Theorems~\ref{thm:firstfrontier}, \ref{thm:secondfrontier} give an alternative way of polytime computing  unique characterisations of $\sQpnd$-queries by  examples \cite{DBLP:conf/kr/FortinKRSWZ22}, opening another route to studying unique characterisations and exact learning of temporal queries.

\section{Complexity}
\label{sec:complexity}

Now we show complexity bounds for the decision and counting problems from Section~\ref{sec:example}, starting with verification and observing that polytime computable $n$-bounded frontiers imply tractable verification of most specific/general separators:

\begin{lemma}\label{lem:cons}
If an $n$-bounded strengthening/weakening frontier for $\q \in \Q$ is
polytime computable in $|\q|$ and $n$, then most specific/general separator verification for $\Q$ is in \PTime.
%
\end{lemma}
\begin{proof}
Let $n= \min \{\max \D \mid \D \in E^+\}$. We compute an $n$-bounded strengthening frontier $\fr$ for $\q$ in $\Q$ and use that $\q \in \sepms(E,\Q)$ iff $\q \in \sep(E,\Q)$ and $\q' \notin \sep(E,\Q)$ for all $\q' \in \fr$. The case of $\sepmg(E,\Q)$ is similar.
\end{proof}
%
\begin{corollary}\label{cor:most}
Most specific/general separator verification is in \PTime{} for $\Qpd$,
$\Qpnd$, and $\Qint$. Most general separator verification is in \PTime{} for $\Qd$.
\end{corollary}
\begin{proof}
	It follows from Lemma~\ref{lem:cons} and Theorems~\ref{thm:firstfrontier}, \ref{thm:secondfrontier}, \ref{thm:ninterval} that most specific/general separator  verification is in \PTime{} for $\Qpd$, $\Qpnd$, and $\Qint$. By Theorem~\ref{thm:thirdfrontier}, most general separation verification is also in \PTime{} for $\Qd$.
\end{proof}

Most specific separator verification is harder for
$\Qd$:

\begin{theorem}\label{thm:mostspeciverif}
  For $\Qd$, the most specific separator verification problem coincides with the unique most specific separator verification problem and is \coNP-complete.
\end{theorem}
\begin{proof}
As we know, $E$ has a unique most specific separator in $\Qd$ iff $\sep(E,\Qd)\ne \emptyset$.
Hence most specific separator verification coincides with unique most specific separator verification.
Given $\q$ and $E$, we can check in \NP{} that either $\q \notin \sep(E,\Qd)$ (which is in \PTime{}) or that there is $\q' \in \sep(E,\Qd)$ with $\q'\models \q$ and $\q'\not\equiv \q$ (which is in \NP). This gives the \coNP{} upper bound. The more involved proof of the lower one (in the appendix) is based on ideas similar to those in the proof of Theorem~\ref{lem:consnew} below.
\end{proof}

In the cases when most specific/general separators are not necessarily unique, we obtain the following:

\begin{theorem}\label{thm:uniqueveri}
Unique most specific/general separator verification in $\Qpnd$, $\Qpd$, $\Qint$ as well as unique most general separator verification in $\Qd$ are \coNP-complete.
\end{theorem}
\begin{proof}
The upper bounds follow from Corollary~\ref{cor:most}.
The lower ones are by reduction of the problem to decide whether, for a Boolean formula $\varphi$ and a satisfying assignment $\mathfrak a$,  there is a satisfying assignment for $\varphi$ different from $\mathfrak a$. We use ideas similar to those in the proof of Theorem~\ref{lem:consnew}.
\end{proof}

We now turn to the existence and counting problems. Recall that \textsc{Unique SAT} is the problem to decide whether there is exactly
one satisfying assignment for a propositional formula~\cite{DBLP:journals/iandc/BlassG82}.
It is in $\Delta_{2}^{p}$, \coNP-hard, and complete for the class US of problems solvable by a non-deterministic polynomial-time Turing machine $M$ that accepts iff $M$ has exactly one accepting path. A counting problem is in the class $\sharp \PTime$ if there is such $M$ whose number of accepting paths coincides with the number of solutions to the problem~\cite{DBLP:books/daglib/0023084}.


\begin{theorem}\label{thm:main}
For $\Qpnd$, $\Qpd$, $\Qint$, counting most speci\-fic/general separators is $\sharp P$-complete\textup{;} the existence of a unique most specific/general sepa\-ra\-tor is \textup{US}-complete. The same  holds for most general separators in $\Qd$.
\end{theorem}

Theorem~\ref{thm:main} follows from the very general and fine-grained complexity results provided by Theorems~\ref{lem:cons2} and \ref{lem:consnew} below.

\begin{theorem}\label{lem:cons2}
Let $\Q \subseteq \Qpnd$ be a class with polytime decidable membership and polytime computable $n$-bounded strengthening/weakening frontiers for queries in $\Q$. Then
%
\begin{itemize}\itemsep=0pt
\item[--] counting most specific/general $\mathcal{Q}$-separators is in $\sharp P$\textup{;}

\item[--] the existence of a unique most specific/general $\mathcal{Q}$-sepa\-ra\-tor is in \textup{US}.
\end{itemize}
This also holds for most general separators in classes of conjunctions of $\nxt\Diamond$-path queries closed under dropping con\-juncts and having polytime computable weakening frontiers.
\end{theorem}
\begin{proof}
Given $E$, construct a TM $M$ that guesses a $\q\in \Q$ with $\dep(\q) \leq n = \min \{\max \D \mid \D \in E^+\}$ and accepts if $\q$ separates $E$ and no $\q'$ in the $n$-bounded frontier for $\q$ in $\Q$ separates $E$. As we know, the required checks are in \PTime{}. In the second claim, $M$ guesses a query $\q = \rho \wedge \Diamond \q_{1} \wedge \dots \land \Diamond \q_{l}$ with $\q_i \in \Qpnd$ and $l\leq |E^{-}|$.
\end{proof}

For the lower bounds, we take into account the boundedness of $\sig(E)$ and the cardinalities $|E^+|$ and $|E^-|$ of positive and negative examples in $E$. The next result provides matching lower bounds for Theorem~\ref{thm:main} even if the signature is bounded and only one of $|E^+|$ and $|E^-|$ is unbounded, except for $\Qint$ and $\Qd$, where $|E^+|$ has to be unbounded.

\begin{theorem}\label{lem:consnew}
Let $\Q \subseteq \Qpnd$ be any class of queries containing all $\Diamond\rho$ with a conjunction of atoms $\rho$. Then
\begin{itemize}\itemsep=0pt
\item[--] counting most specific/general $\mathcal{Q}$-sepa\-ra\-tors is $\sharp \PTime$-hard\textup{;}

\item[--] the existence of unique most specific/general separator is \textup{US}-hard
\end{itemize}
even for $E = (E^+,E^-)$ and $\sigma = \sig(E)$ with
$(a)$ $|E^{-}|\leq 1$ or $(b)$ $|E^{-}| \leq 1$, $|\sigma| = 2$ and
$\mathcal{Q}=\Qint$ or $\mathcal{Q}\supseteq \Qpd$, or $(c)$ $|E^{+}| \leq 4$, $|\sigma| = 3$, $\mathcal{Q}\supseteq \Qpd$.
This result also holds for most general separators in $\Qd$ and $|E^{-}|=1$, $|\sigma|=2$.
\end{theorem}
\begin{proof}
We sketch the proof of the first claim by a parsimonious reduction from SAT for an unbounded number of negative examples (which can be easily merged into one). Take a CNF $\varphi=\psi_1\wedge \dots \wedge \psi_k$ with clauses $\psi_i$ over variables $x_1, \dots, x_n$. 
We construct $E=(E^{+},E^{-})$ such that there is a bijection between the satisfying assignments for $\varphi$ and the separators for $E$ in $\Qnd$ (even queries of the form $\Diamond\rho$). The claim follows from the fact that the separating queries are mutually $\models$-incomparable. Define the positive examples $E^+ = \{\D_0, \D_0',\D_1,\dots,\D_n\}$ with $2n$ atoms $A_1,\bar A_1, \dots, A_n,\bar A_n$ by taking
\begin{align*}
& \D_0 = \{A_i(1),\bar A_i(1)\mid 1\leq i\leq n\},\\
& \D_0'=\{A_i(2), \bar A_i(2)\mid 1\leq i\leq n\},\\
& \D_i = \{A_i(1), \bar A_i(2), A_j(1),\bar A_j(1),A_j(2),\bar A_j(2) \mid i \ne j\},
\end{align*}
for $1\leq i\leq n$. Let $E^{-}=\{\D_1^1,\dots,\D_k^1,\D^2_1,\dots,\D^2_n\}$, where $\D_i^1$, $1\leq i\leq k$, comprises $\bar A_j(1)$ if $x_j$ does not occur negatively in $\psi_i$, and $A_j(1)$ if $x_j$ does not occur positively in $\psi_i$, and
$\D_i^2 = \{A_j(1),\bar A_j(1) \mid j\neq i\}$.
For an assignment $\mathfrak a$ for $x_1,\dots,x_n$, let $\rho_{\mathfrak a}$ contain $A_i$ if $\mathfrak a(x_i)=1$ and $\bar A_i$, otherwise. Now our claim follows from the following: 
$(i)$ if $\mathfrak a$ satisfies $\varphi$, then $\Diamond \rho_{\mathfrak a}$ separates $E$; $(ii)$ if $\q\in\Qpnd$ separates $E$, then $\q\equiv\Diamond \rho_{\mathfrak a}$, for some $\mathfrak a$ satisfying $\varphi$.		

To bound $\sigma$ and/or $E^+$, we give  parsimonious reductions from SAT and employ techniques for the longest common subsequence problem and separator existence for sequence queries in $\Qpd$~\cite{DBLP:conf/cpm/BlinBJTV12,DBLP:journals/tcs/Fraser96}.
\end{proof}

Again, re-using techniques
from the literature on algorithms for sequences needs some care.
Similar to separability, for example sets of sequences, counting most general/specific sequence $\Qint$-separators is easily seen to be in $\PTime$ even for unbounded signatures and example sets, but Theorem~\ref{lem:consnew}
shows that this is not so for $\Qint$ on non-sequence data instances.

Surprisingly, the complexities of counting most general separators and deciding the existence of a unique one diverge if we bound the number of examples, cf.\ Theorems~\ref{thm:algorithms} and~\ref{forDiamond}.

\begin{theorem}\label{lem:consnew+}
Let $\Q=\Qd$ or $\Q \subseteq \Qpnd$ be any class containing  all $\Diamond\rho$ with a conjunction of atoms $\rho$. Then counting most general $\mathcal{Q}$-sepa\-ra\-tors is $\sharp \PTime$-hard even for example sets $E = (E^+,E^-)$ with $|E^{+}|=2$ and $|E^{-}|=1$.
\end{theorem}
\begin{proof} The proof is by reduction from counting satisfying assignments
  for monotone formulas:
  $E^{-}$ is as in the proof of Theorem~\ref{lem:consnew} and the positive examples are $\D_{0}$, $\D_{0}'$.
\end{proof}
Theorem~\ref{lem:consnew+} does not hold for counting
most specific separators in $\Qint$, which is easily seen to be in
$\PTime$ if $|E^+|$ is bounded. It remains open whether some version of
this theorem holds for most specific separators in
$\Qpnd$ or $\Qpd$.

\newcommand{\tp}{t_+}
\newcommand{\tm}{t_-}
\newcommand{\cp}{c_+}
\newcommand{\cm}{c_-}
\newcommand{\tpcp}{t_+^{c_+}}
\newcommand{\tmcm}{t_-^{c_-}}

\section{Algorithms}\label{algs}

The frontiers defined in Section~\ref{sec:frontiers} give a polytime algorithm for computing a most specific/general separator starting from \emph{any} given separator. Suppose $\Q\in \{\Qpd, \Qpnd,\Qint\}$, we are given $\q \in \sep(E,\Q)$ and need a most specific separator. By  Theorems~\ref{thm:firstfrontier}, \ref{thm:secondfrontier}, the length of the longest $\models$-chain in $\sep(E,\Qpnd)$ is polynomial in $|E|$.\footnote{In contrast, the proof of Theorem~\ref{lem:consnew} shows the size of the maximal antichain in $\sep(E,\Qpd)$ is in general exponential in $|E|$.} We take, if possible, some $\q'\in \sep(E,\Q)$ in the strengthening frontier for $\q$, then $\q'' \in \sep(E,\Q)$ in the strengthening frontier for $\q'$, etc. This process terminates after polynomially-many steps, returning a most specific $\Q$-separator (and so the unique one, if any).

Thus, we can focus on algorithms deciding the existence of a (unique most specific/general) separator and constructing it. In the next theorem, we compute not just a random input separator, but a longest/shortest one. As strengthening/weakening frontiers contain queries that are not shorter/longer than the input, the procedure above will compute a longest most specific/shortest most general separator. 


\begin{theorem}\label{thm:algorithms}
Let $E = (E^+,E^-)$, $\sigma = \sig(E)$, $\tp$/$\tm$ be the maximum timestamp in $E^+$/$E^-$, $\cp = |E^+|$, $\cm = |E^-|$, and $\Q \in\{\Qpd, \Qpnd\}$.
The following can be done in time $O(\tpcp \tmcm)$\textup{:}
\begin{itemize}\itemsep=0pt
\item[$(a)$] deciding whether $\sep(E,\Q)\ne \emptyset$\textup{;}

\item[$(b)$] computing a longest/shortest separator in $\sep(E,\Q)$\textup{;}

\item[$(c)$] deciding the existence of a unique most specific $\Q$-separator and a unique most general $\Qpd$-separator, and constructing such a separator.
\end{itemize}
For bounded $\cp$ and $\cm$, problem $(a)$ is in $\NL$.
\end{theorem}
\begin{proof}
We only sketch the construction for $\Q = \Qpd$.

$(a)$ First, we define a directed labelled rooted graph $\mathfrak P$ whose paths from the root represent $\Qpd$-queries with positive examples $E^+ = \{\D_1, \dots, \D_{\cp}\}$. Its nodes are vectors $(n_1, \dots, n_{\cp})$ with $0 \leq n_i \leq \max \D_i$, which are labelled by the sets $\{A \in \sigma \mid A(n_i) \in \D_i \text{ for all $i$}\}$, and the edges are $(n_1, \dots, n_{\cp}) \to (n_1', \dots, n'_{\cp})$ with $n_i < n_i'$ for all $i$. The \emph{root} of $\mathfrak P$ is $\bar 0 = (0, \dots, 0)$.
To illustrate, let $E = (E^+,E^-)$, where $E^+ = \{\D_1^+, \D_2^+ \}$, $E^- = \{\D_1^-, \D_2^-, \D_3^- \}$,
\begin{align*}
& \D_1^+ = \{ A(0), C(1), D(1), B(2) \}, \\
& \hspace*{3.7cm} \D_2^+ = \{ A(0), C(1), B(2), D(2) \},\\
& \D_1^- = \{ A(0), C(1) \}, \ \D_2^- = \{ A(0), D(1) \}, \ \D_3^-= \{ B(0)\}.
\end{align*}
Graph $\mathfrak P$ is shown on the left-hand side of the picture below: \\
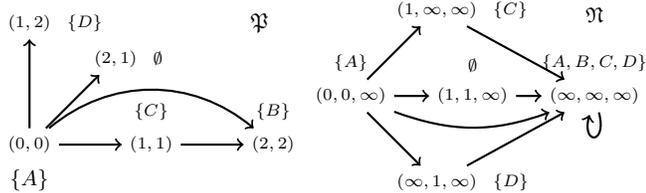
\centerline{
\begin{tikzpicture}[->,thick,node distance=1.8cm, transform shape,scale=0.9]\footnotesize

\node[rectangle, label=below:{$\{A\}$}] (s00) {\scriptsize$(0,0)$};
\node[rectangle, right of = s00, label=above:{\scriptsize$\{C\}$}] (s11) {\scriptsize$(1,1)$} edge[above, <-] (s00);
\node[rectangle, above right of = s00, label=right:{\scriptsize$\emptyset$}] (s21) {\scriptsize$(2,1)$} edge[left, <-] (s00);
\node[rectangle, above of = s00, label=right:{\scriptsize$\{D\}$}] (s12) {\scriptsize$(1,2)$} edge[left, <-] (s00);
\node[rectangle, right of = s11, label=above:{\scriptsize$\{B\}$}] (s22) {\scriptsize$(2,2)$} edge[above, <-] (s11) edge[<-, out = 140, in =40] (s00);

\node at (3.4,1.8) {$\mathfrak P$};
\end{tikzpicture}
\begin{tikzpicture}[->,thick,node distance=1.8cm, transform shape,scale=0.9]\footnotesize
\node[rectangle, label=above:{\scriptsize$\{A\}$}] (s00) {\scriptsize$(0,0, \infty)$};
\node[rectangle, right of = s00, label=above:{\scriptsize$\emptyset$}] (s11) {\scriptsize$(1,1, \infty)$} edge[above, <-] (s00);
\node[rectangle, above right of = s00, label=right:{\scriptsize$\{C\}$}] (s1i) {\scriptsize$(1,\infty, \infty)$} edge[left, <-] (s00);
\node[rectangle, below right of = s00, label=right:{\scriptsize$\{D\}$}] (si1) {\scriptsize$(\infty,1,\infty)$} edge[left, <-] (s00);
\node[rectangle, right of = s11, label=above:{\scriptsize$\{A,B,C,D\}$}] (sii) {\scriptsize$(\infty, \infty, \infty)$} edge[<-] (s11) edge[<-, out = 200, in =-20] (s00) edge[<-] (si1) edge[<-] (s1i) edge[loop below] (sii);

\node at (3.6,1.2) {$\mathfrak N$};
\end{tikzpicture}
}\\
Each path starting at $\bar 0$ gives rise to a $\sQpd$-query with positive examples in $E^+$: e.g., $(0, 0) \to (1, 1) \to (2, 2)$ gives rise to the query $A \land \Diamond (C \land \Diamond B)$.

Next, define another graph $\mathfrak{N}$ for $E^- = \{\D_1, \dots, \D_{\cm}\}$. Its nodes are $(n_1, \dots, n_{\cm})$, where $n_i \in [0, \max \D_i] \cup \{\infty\}$, including $\bar \infty = (\infty, \dots, \infty)$, which is labelled by $\sigma$. The label of $(n_1, \dots, n_{\cm})$ is $\{A \mid A(n_i) \in \D_i, \text{ for all } n_i \neq \infty\}$. The edges are defined in the same way as for $\mathfrak P$, with $n_i < \infty$, for any $n_i$. The root of $\mathfrak{N}$ is $(n_1, \dots, n_l)$, where $n_i = 0$ if $\{A(0) \mid A(0) \in \D, \text{ for all } \D \in E^+\} \subseteq \D_i$ and $n_i = \infty$ otherwise (see the picture above).
Let $\q$ be a $\sQpd$-query of the form~\eqref{dnpath} with $\rho_0$ contained in the root's label. Then all $\D_i \in E^-$ are negative examples for $\q$ iff every path starting at the root and having labels $\rho'_0,\dots,\rho'_n$ with $\rho'_i \supseteq \rho_i$, for all $i \le n$, comes through node $\bar\infty$. In our example, every such path for $\q = A \land \Diamond B$ and $\q = A \land \Diamond (C \land \Diamond B)$ involves $\bar\infty$. However, this is not the case for $\q = A \land \Diamond C$.

Consider now a graph $\mathfrak P \otimes \mathfrak{N}$ with nodes $(\avec{n}, \avec{m})$, where $\avec{n}$ is a node in $\mathfrak P$ and $\avec{m}$ a node in $\mathfrak N$ with $\avec{l}(\avec{n}) \subseteq \avec{l}(\avec{m})$, for the labels $\avec{l}(\avec{n})$ and $\avec{l}(\avec{m})$ of $\avec{n}$ and $\avec{m}$. We have an edge $(\avec{n}, \avec{m}) \to (\avec{n}', \avec{m}')$ in $\mathfrak P \otimes \mathfrak{N}$ iff $\avec{n} \to \avec{n}'$ in $\mathfrak P$, $\avec{m} \to \avec{m}'$ in $\mathfrak N$, and $\mathfrak P \otimes \mathfrak{N}$ has no $(\avec{n}', \avec{m}'')$ with $\avec{m} \to \avec{m}''$, $\avec{m}''_i < \avec{m}'_i$ and $\avec{m}_i'' \neq \infty$, for some $i$, $\avec{m}_i$ being the $i$th coordinate of $\avec{m}$. One can see that, for any $(\avec{n}, \avec{m})$ in $\mathfrak P \times \mathfrak N$ and $\avec{n}'$ in $\mathfrak P$, there exists at most one edge $(\avec{n}, \avec{m}) \to (\avec{n}', \avec{m}')$. The root of $\mathfrak P \times \mathfrak N$ comprises the roots  of $\mathfrak P$ and $\mathfrak N$. In our example, the edges from the root $(\bar 0, (0,0, \infty))$ of $\mathfrak P \otimes \mathfrak N$ lead to $((1,2), (\infty,1,\infty))$, $((2,1), (1,1,\infty))$,  $((1,1), (1,\infty,\infty))$ and $((2,2), \bar{\infty})$.
Given a path
\begin{equation}\label{eq:prod-path}
\pi = (\bar 0 =\avec{n}_0, \avec{m}_0),\dots,(\avec{n}_n, \avec{m}_n)
\end{equation}
let $\q_\pi$ be the $\Qpd$-query of the form~\eqref{dnpath} with $\rho_i = \avec{l}(\avec{n}_i)$ (note that $\q_\pi$ is not necessarily in normal form). We call $\pi$ a \emph{separating path} for $E$ if $\avec{l}(\avec{n}_n) \ne \emptyset$ and $\avec{m}_n = \infty$.

\begin{lemma}\label{lem:sep-bounded-diamond-prod}
$\sep(E,\Q) \ne \emptyset$ iff $\mathfrak P \otimes \mathfrak N$ contains a separating path $\pi$, with $\q_\pi$ separating $E$.
\end{lemma}

In our running example, $\mathfrak P \otimes \mathfrak N$ has two \emph{separating paths}: $\pi_1 = (\bar 0, (0,0, \infty)), \, ((2,2), \bar \infty)$ and
$\pi_2 = (\bar 0, (0,0, \infty)),$ $((1,1), (1,\infty,\infty)),\, ((2,2), \bar \infty)$, which give rise to the separators $\q_{\pi_1} = A \land \Diamond B$ and $\q_{\pi_2} = A \land \Diamond (C \land \Diamond B)$.

The existence of a separating path in $\mathfrak P \otimes \mathfrak N$ can be checked in time $O(\tpcp \tmcm)$. If $\cp$ and $\cm$ are bounded, given $(\avec{n}, \avec{m})$ and $(\avec{n}', \avec{m}')$, we can check in logspace whether $(\avec{n}, \avec{m})$ is the root of $\mathfrak P \otimes \mathfrak N$ and $(\avec{n}, \avec{m}) \to (\avec{n}', \avec{m}')$, and so the existence of a separating path can be decided in \NL.
%

The proof of point $(b)$ relies on the following observation:
\begin{lemma}\label{th:bounded-subs}
If $\q$ of the form~\eqref{dnpath} separates $E$, then there is a separating path of the form~\eqref{eq:prod-path} with $\rho_i \subseteq \avec{l}(\avec{n}_i)$, for all $i \le n$.
\end{lemma}

It follows that the length of a longest/shortest separator for $E$ coincides with the length of a longest/shortest separating path in $\mathfrak P \otimes \mathfrak N$, which can be found in polytime.

The proof of point $(c)$ is based on the following criterion:

\begin{lemma}\label{th:strongest-crit}
A $\Qpd$-query $\q$ is a unique most specific $\Qpd$-separator for $E$ iff there is a separating path $\pi$ such that $\q = \q_\pi$ and $\q_\pi \models \q_\nu$, for every separating path $\nu$.
\end{lemma}

In our example, $\q_{\pi_2}$ is a unique most specific separator. We show that the criterion of Lemma~\ref{th:strongest-crit} can be checked in polytime in $\mathfrak P \times \mathfrak N$.
For unique most general separators, the seemingly obvious inversion of $\q_\pi \models \q_\nu$ does not give a criterion, and a different type of graph is required.
%
\end{proof}

Finally, we show how Theorem~\ref{thm:algorithms} can be used to check the existence of and construct unique most general separators in $\Qd$ (the case of unique most specific ones is trivial).

\begin{theorem}\label{forDiamond}
An example set $E=(E^{+},E^{-})$ has a unique most general separator in $\Qd$ iff $\q=\bigwedge_{\D \in S} \q_{\D}$ separates $E$, where $S$ is the set of $\D\in E^{-}$ such that $(E^+,\{\D\})$ has a unique most general separator, $\q_{\D}$, in $\Qpd$. In this case, $\q$ is a unique most general separator of $E$ in $\Qd$.
\end{theorem}

\section{Conclusions}
\label{sec:conc}
We have conducted a comprehensive complexity analysis of extremal
separators in the spaces $\sep(E,\mathcal{Q})$ with $\Q$ ranging from
various classes of temporal path $\nxt\Diamond$-queries to arbitrary
$\Diamond$-queries.
For arbitrary $\nxt\Diamond$-queries, we only know more or less
straightforward upper bounds such as $\Pi_{2}^{p}$ for (unique) most
general separator verification and $\Sigma_{3}^{p}$ for unique most
general separator existence. Establishing tight bounds remains a
challenging open problem, which requires a deeper understanding of
query containment for these queries.

We also plan to analyse the shortest and longest separators. For
instance, we show in the appendix that verifying such separators in
$\sep(E,\Qpnd)$ is \coNP-complete---harder than verifying most
specific/general ones. An empirical evaluation of our algorithms and
more expressive query languages (say, with non-strict $\Diamond$ and
`until') are left for future work.


\newpage

\bibliographystyle{named}


\newpage

\begin{appendix}
\onecolumn
\noindent{\Large \bf Appendix: Proofs}

\section{Proofs for Section~\ref{sec:cont}}
    We assume that queries $\q$ are in normal form. Recall that $[n]=\{0,\ldots,n\}$. We start by introducing a helpful 
    tool for checking whether $\D,0\models \q$.
	Let $\q$ be a path $\nxt\Diamond$-query of the form \eqref{dnpath}, that is,
	\begin{equation*}
		\q = \rho_0 \land \op_1 (\rho_1 \land \op_2 (\rho_2 \land \dots \land \op_n \rho_n) ),
	\end{equation*}
	where $\op_i \in \{\nxt, \Diamond \}$ and $\rho_i$ is a conjunction of atoms.
	A \emph{satisfying function} $f$ for $\q$ into a data instance $\D$ maps $[n]$ to $[\max\D]$ such that
	\begin{itemize}
		\item $f(0)=0$;
		\item $\rho_{i} \subseteq \D_{f(i)}:=\{ A \mid A(f(i))\in \D\}$;
		\item $f(i)<f(i+1)$;
		\item $f(i+1)=f(i)+1$ if $\op_{i+1}=\nxt$.
	\end{itemize}
It is easy to show that $\D,0\models \q$ iff there exists a satisfying function for $\q$ into $\D$.
\begin{eqnarray}\label{defofq}
\q'= \rho_{0}' \wedge \op_{1}' (\rho_{1}' \wedge \cdots \wedge\op_{m}' (\rho_{m}')).
\end{eqnarray}
A function $g:[m] \rightarrow [n]$ is a \emph{containment witness} for $\q,\q'$
if the following conditions holds:
\begin{itemize}
	\item $\rho_{i}' \subseteq \rho_{g(i)}$;
	\item $g(0)=0$;
	\item $g(i)<g(i+1)$;
	\item if $\op_{i+1}'=\nxt$, then $\op_{g(i+1)}=\nxt$ and $g(i+1)=g(i)+1$.
\end{itemize}
We can now state Lemma~\ref{lem:containment} as follows.
	\begin{lemma}\label{lem:cont}
		$\q \models \q'$ iff there is a containment witness for $\q,\q'$.
	\end{lemma}
	\begin{proof}
		It is easy to see that if there is a containment witness for $\q,\q'$ then $\q\models \q'$. We prove the converse direction.
				
		Let $\sigma=\sig(\q)$. The proof is by induction on the temporal depth of $\q'$. If $\dep(\q')=0$, the claim is trivial.
		
		Assume $\q\models \q'$ are given and the statement has been shown for all $\q''$ with $\dep(\q'')<\dep(\q')$.
		
		Assume first that $\q'$ starts with a nonempty sequence of $\nxt$.
		Then let $k>0$ be maximal such that $\op_{1}'=\cdots = \op_{k}'=\nxt$.
		Then $\rho_{k}'\not=\top$ (by normal form) and if $m>k$, then $\op_{k+1}'=\Diamond$.
		
		Let $r$ be minimal such that
		$$
		\rho_{0} \wedge \op_{1} (\rho_{1} \wedge \cdots \wedge\op_{r}\rho_{r}))\models \rho_{0}' \wedge \op_{1}' (\rho_{1}' \wedge \cdots \wedge\op_{k}'\rho_{k}'))
		$$
		
		\medskip
		\noindent
		\emph{Claim 1.} For all $i\leq k$, $\op_{i}=\nxt$ and $\rho_{i}\supseteq \rho_{i}'$. Hence $r=k$.
		
		\medskip
		\noindent 
		To prove Claim~1 observe that if there is $i\leq k$ with $\op_{i}=\Diamond$, then
		we have for
		$$
		\D = \rho_{0}\ldots\rho_{i-1}\emptyset^{m}\rho_{i}\ldots\rho_{n}
		$$
		$\D,0\models \q$ and $\D,0\not\models\q'$, a contradiction.
		If there is $i\leq k$
		with $\rho_{i}\not\supseteq \rho_{i}'$, then we have for
		$$
		\D = \rho_{0}\ldots\rho_{n}
		$$
		$\D,0\models \q$ and $\D,0\not\models\q'$, again a contradiction. This finishes the proof of Claim~1.
		
		Next observe that if follows from $\q\models \q'$ that
		$$
		\delta \models \delta'
		$$
		for $\delta=\op_{k+1}(\rho_{k+1}\wedge \cdots \wedge \op_{n}\rho_{n}))$ and $\delta'=\op_{k+1}'(\rho_{k+1}' \wedge \cdots \op_{m}' \rho_{m}'))$.
		As
		$$
		\dep(\delta')< \dep(\q')
		$$
		we can apply the IH to $\delta, \delta'$ and have a containment witness $g'$ for $\delta,\delta'$. Define $g$ by setting $g(i)=i$ for all $i\leq k$
		and $g(k+j)=g'(j)+k$ for all $j$ with $1\leq j \leq m-k$.
		Then $g$ is a containment witness for $\q,\q'$.		
		
		\medskip
		
		Assume now that $\q'$ does not start with $\nxt$.
		Assume first there is a minimal initial subquery
		$$
		\delta'=\rho_{0}' \wedge \op_{1}'(\rho_{1}'\wedge \cdots \wedge \op_{k}'\rho_{k}'))
		$$
		of $\q'$ ending with $\rho_{k}'\not=\top$
		and, moreover, $\op_{k+1}'=\Diamond$ and $k>0$.
		Let
        $$
        \delta = \rho_{0} \wedge \op_{1}(\rho_{1}\wedge \cdots \wedge \rho_{r}))
        $$
        be the minimal initial subquery of $\q$ such that $\delta\models \delta'$. As $\dep(\delta')< \dep(\q')$, by IH we have a containment witness $g_{0}$ for $\delta,\delta'$. As $\rho_{k}'\not=\top$ and $r$ is minimal we have $g_{0}(k)=r$. 
        By IH the claim also holds for
        $\gamma'=\op_{k+1}'(\rho_{k+1}' \wedge \cdots \wedge \op_{m}'\rho_{m}'))
        $
        and
        $
        \gamma=\op_{r+1}(\rho_{r+1} \wedge \cdots \wedge \op_{n}\rho_{n}))$.
        Hence we obtain a containment witness $g_{1}$ for $\gamma,\gamma'$. Then we obtain a containment witness for $\q,\q'$ by concatenating $g_{0}$ and $g_{1}$.

        Finally assume there is no $\rho_{k}'\not=\top$
        with $op_{k+1}'=\Diamond$ and $k>0$. Then, by normal form, $\q'$ takes the form
        $$
        \rho_{0}'\wedge \Diamond^{k}(\rho_{k}' \wedge \op_{k+1}'\rho_{k+1}'\cdots \wedge \op_{m}'\rho_{m}'))
        $$
        with $k>0$, $\rho_{k}'\not=\top$, and $\op_{k+1}'=\cdots =\op_{m}'=\nxt$.

        We show there exists $j\geq k$ such that
        $\op_{j+1}=\cdots =\op_{j+m-k}=\nxt$ and $\rho_{j+\ell} \supseteq \rho_{k+\ell}'$ for all $\ell$ with $0\leq \ell\leq m-k$.
        Clearly then there is a containment witness $g$ for $\q,\q'$.
	
	    To show our claim, suppose no such $j$ exists. Then take the data instance $$
	    \D = \rho_{0}w_{1}\rho_{1}\cdots w_{n}\rho_{n}
	    $$
	    where $w_{i}$ is the empty word if $\op_{i}=\nxt$ and $w_{i}=\emptyset^{m}$ if $\op_{i}=\Diamond$. Then we have $\D,0\models \q$ but $\D,0\not\models \q'$, a contradiction.
	\end{proof}

We consider $\Qd$. We assume that $\q\in \Qd$
takes the from
$$
\rho \wedge \q_{1} \wedge \cdots \wedge \q_{n}
$$
where $\rho$ is a conjunction of atoms and each $\q_{i}$ is in $\Qpd$ and starts with $\Diamond$.

\medskip
\noindent
{\bf Lemma~\ref{lem:decomp}.}
	\emph{If $\q=\rho \wedge \q_{1} \wedge \cdots \wedge \q_{n}\in \Qd$ is in normal form, $\q'\in \Qpd$ and $\q\models \q'$, then there is $\q_i$, $1 \le i \le n$, with $\rho \wedge \q_{i} \models \q'$.}

\medskip
\noindent
\begin{proof}
	Assume
	$$
	\q_{i}=\Diamond (\rho^{1}_{i} \wedge \cdots \rho^{n_{i}}_{i}))
	$$
	for $1\leq i \leq n$,
	$$
	\q'= \rho^{0} \wedge \Diamond (\rho^{1} \wedge \cdots \wedge \Diamond \rho^{m})),
	$$
	and $\q\models \q'$.
	Note that $\rho\supseteq \rho^{0}$ as otherwise $\q\not\models\q'$. We define functions
	$$
	f_{i}: \{1,\ldots,m\} \rightarrow \{1,\ldots, n_{i}\}\cup \{\infty\}
	$$
	for $1\leq i \leq n$. The definition is by induction starting from $1$.
	
	Set $f_{i}(1)=j$ if $j$ is minimal such that $\rho_{i}^{j}\supseteq \rho^{1}$ and $f_{i}(1)=\infty$ if no $j$ with $\rho_{i}^{j}\supseteq \rho^{1}$ exists.
	
	Inductively, if $f_{i}(\ell)=\infty$, then $f_{i}(\ell+1)=\infty$.
	If $f_{i}(\ell)<\infty$, then set $f_{i}(\ell+1)=j$ if $j$ is minimal such that
	\begin{itemize}
		\item $j>f_{i}(\ell)$;
		\item $\rho_{i}^{j}\supseteq \rho^{\ell+1}$
	\end{itemize}
	and $f_{i}(\ell+1)=\infty$ if no $j>f_{i}(\ell)$ such that $\rho_{i}^{j}\supseteq \rho^{\ell+1}$ exists.
	
	Observe that if there exists $i\leq n$ such that $f_{i}(m)<\infty$, then
	$\rho \wedge \q_{i}\models \q'$, as required. Hence it remains to show that there exists such an $i\leq n$. For a proof by contradiction, assume there is no such $i\leq n$. We derive a contraction by proving that $\q\not\models \q'$. Let $m'\leq m$ be minimal such that $f_{i}(m')=\infty$ for all $i\leq n$. Define a data instance $\D_{1}$ as
	$$
	\D_{1}=\rho\rho_{1}^{1}\cdots \rho_{1}^{k_{1}}\cdots\rho_{n}^{1}\cdots \rho_{n}^{k_{n}},
	$$
	where $k_{i}=\min\{n_{i},f_{i}(1)-1\}$ for $1\leq i \leq n$ (we set $\infty-1=\infty$). If $m'=1$, then $\D_{1},0\models \q$ and $\D_{1},0\not\models\q'$ since $\D_{1},0\not\models\Diamond\rho^{1}$, and we are done. Otherwise let
	$$
	\delta_{1}= \bigcup_{1\leq i \leq n,f_{i}(1)<\infty}\rho_{i}^{f_{i}(1)}
	$$
	We continue in this way for $\ell=2,3,\ldots,m'$ as follows.
	First define for $2\leq \ell \leq m'$:
	$$
	\D_{\ell}=
	\rho_{1}^{f_{1}(\ell)+1}\cdots \rho_{1}^{k_{\ell,1}}\cdots
	\rho_{n}^{f_{n}(\ell)+1}\cdots \rho_{n}^{k_{\ell,n}},
	$$
	where $k_{\ell,i}=\min\{n_{i},f_{i}(\ell+1)-1\}$ for $1\leq i \leq n$ (note that $\rho_{i}^{f_{i}(\ell)+1}\cdots \rho_{i}^{k_{\ell,i}}$ is empty if $f_{i}(\ell)+1 >n_{i}$). Then let for $2\leq \ell <m'$:
	$$
	\delta_{\ell}= \bigcup_{1\leq i \leq n,f_{i}(\ell)<\infty}\rho_{i}^{f_{i}(\ell)}
	$$
	and set
	$$
	\D = \D_{1}\delta_{1}\D_{2}\delta_{2}\cdots\delta_{m'-1}\D_{m'}
	$$
	It is straightforward to show that $\D,0\models \q$ and
	$\D,0\not\models \q'$. Hence $\q\not\models\q'$ and we have derived a contradiction.
\end{proof}
\section{Proofs for Section~\ref{sec:example}}
We prove the complexity results for separability that are stated in this section for $\Qint$.
\begin{theorem}\label{thm:qinappendix}
	(1) For $\Qint$, separator existence is \NP-complete. This also holds if $|E^{-}|$ and $\sig(E)$ are bounded.
	
	(2) If $|E^{+}|$ is bounded then $\Qint$-separator existence is in \PTime.
	
	(3) For example data instances $E=(E^{+},E^{-})$ that are sequences, the existence of a separator that is a sequence query in $\Qint$ is in \PTime.
\end{theorem}
\begin{proof}
	(1) The $\NP$-upper bound follows from the fact that given $E$ one can guess a separator of polynomial size in $\Qint$ and check in polytime that is separates. The lower bound follows directly from the 
	proof of Theorem~\ref{lem:consnew} (b) given later in this appendix.
	
	(2) Assume $E^{+}=\{\D_{1},\ldots,\D_{n}\}$. Let $m_{i}=\max \D_{i}$ and $m=\max \{m_{i} \mid 1\leq i\leq n\}$. Consider for every $0\leq k<m$ and vector $\vec{k}=k_{1},\ldots,k_{n}$ with $0< k_{i}\leq m_{i}$ the query
	$$
	\q_{k,\vec{k}}= \Diamond(\rho_{k,\vec{k}}^{0} \wedge \nxt(\cdots \wedge \nxt \rho_{k,\vec{k}}^{k}))
	$$
	where
	$$
	\rho_{k,\vec{k}}^{j}= \bigcap_{1\leq i \leq n}\{ A | A(k_{i}+j)\in \D_{i}\}
	$$
	Let $\Xi$ denote the set of queries of the form $\q_{k,\vec{k}}$ that are in $\Qint$. Then it is easy to see that there exists a query in $\Qint$ that separates $E$ iff some query in $\Xi$ separates $E$. The latter can be checked in polynomial time if $n$ is fixed.
	
	(3) Let $E^{+}=\{\D_{1},\ldots,\D_{n}\}$ be a set of sequences. Assume 
	$$
	\D_{1}= \{A_{0}\}\{A_{1}\}\cdots \{A_{m}\}
	$$
	Then any possible separating sequence query in $\Qint$ takes the form
	$$
	\Diamond(A_{k_{1}} \wedge \nxt(\cdots \nxt A_{k_{2}}))
	$$
	with $0<k_{1}\leq k_{2}\leq m$. Clearly it can be checked in polynomial time whether any such query separates $E$.
\end{proof}

\section{Proofs for Section~\ref{sec:frontiers}}

{\bf Theorem~\ref{thm:secondfrontier}.}
	{\em Let $\q \in \sQpnd$ be in normal form~\eqref{dnpath}. A weakening frontier for $\q$ in $\sQpnd$ can be computed in polytime by applying once to $\q$ one of the following operations, for $i \in [n]$ and $\q_i = \rho_i \land \op_{i+1}(\rho_{i+1} \land \dots \land \op_n \rho_n)$\textup{:}
	\begin{enumerate}
		\item drop some atom from $\rho_{i}$\textup{;}
		
		\item replace some $\Diamond(\top \wedge \Diamond \q_{i})$ in $\q$ by $\Diamond \q_{i}$\textup{;}
		
		\item replace some $\op_{i}=\nxt$ by $\Diamond$.
	\end{enumerate}
	If $\q \in \sQpd$, a weakening frontier for $\q$ in $\sQpd$ can be constructed in polynomial time using operations 1 and 2.
}
\medskip
\noindent
\begin{proof}
     Let $\mathcal{F}$ denote the set of queries defined in Theorem~\ref{thm:secondfrontier}. It is easy to see using Lemma~\ref{lem:containment} that $\q\models \q'$ and $\q\not\equiv \q'$ for all $\q'\in \mathcal{F}$.
	
	Assume now $\q\models \q'$ and $\q\not\equiv \q'$ for some $\q'\in \sQpnd$.
	We show that there exists a query $\q''\in \mathcal{F}$ such that $\q''\models \q'$. Assume
	$$
	\q'= \rho_{0}' \wedge \op_{1}' (\rho_{1}' \wedge \cdots \wedge\op_{m}' (\rho_{m}')).
	$$
	and assume that $g:[m] \rightarrow [n]$
	is a containment witness for $\q,\q'$.
	
	We distinguish the following cases.
	
	Case 1. $g$ is a surjective mapping onto $[n]$. Then $n=m$ and $g(i)=i$ for all $i\in [m]$. Hence $\rho_{i}'\subseteq \rho_{i}$ for all $i\in [m]$. Then, as $\q\not\equiv \q'$, either $\op_{i}=\nxt$ and $\op_{i}'=\Diamond$ for some $i\in [m]$ or $\rho_{i}'\subsetneq \rho_{i}$ for some $i\in [m]$. In the first case $\q''\models \q'$ for a $\q''\in \mathcal{F}$ obtained in Point~3 and in the second case $\q''\models \q'$ for a $\q''\in \mathcal{F}$ obtained in Point~1.
	
	Case 2. $g$ is not onto $[n]$. Take any $i\leq n$ that is not in the range of $g$. If there is $i<m$ with $f(i+1)-f(i)\geq 2$, then $\q''\models \q'$ for a $\q''\in \mathcal{F}$ obtained in Point~2. Otherwise $n$ is not in the range of $g$ and $\q''\models \q'$ for a $\q''\in \mathcal{F}$ obtained in Point~1.
\end{proof}

{\bf Example~\ref{expDiam}}
	We represent queries in $\sQd$ of the form~\eqref{dnpath} as $\rho_{0}\dots\rho_{n}$. For $\sigma=\{A_{1},A_{2},B_{1}, B_{2}\}$, 
	let $\q_{1}= \emptyset (\s\sigma)^{n} \s$, $\q_{2} = \emptyset\sigma^{2n+1}$, and $\s=\{A_{1},A_{2}\}\{B_{1},B_{2}\}$. We show using~\cite[Example 18]{DBLP:conf/kr/FortinKRSWZ22} that any strengthening frontier for the query $\q_1 \land \q_2$ in $\sQd$ is of size $O(2^n)$.	
	Indeed, let $P$ be the set of queries of the form $\emptyset\s_{1}\dots \s_{n+1}$, where $\s_{i}$ is either $\{A_{1}\}\{A_{2}\}$ or $\{B_{1}\}\{B_{2}\}$. Clearly, $|P| = 2^{n+1}$. As shown in~\cite{DBLP:conf/kr/FortinKRSWZ22}, $\q_{1}\wedge \q_{2}\not\models \q$ for all $\q\in P$ and, for any data instance $\mathcal{D}$ with $\mathcal{D}\models \q_{1}\wedge \q_{2}$, there is at most one $\q\in P$ with $\mathcal{D}\not\models \q$. It follows that $|\fr| \ge 2^{n+1}$, for otherwise the pigeonhole principle would give distinct $\q^1,\q^2 \in P$ and some $\q \in \fr$ such that $\q_1 \land \q_2 \land \q^i \models \q$, $i=1,2$. As $\q_1 \land \q_2 \not\models \q$, there is a data instance $\D \models \q_1 \land \q_2$ with $\D \not\models \q$, and so $\D \not\models \q^i$, for $i=1,2$, contrary to $\q^1 \ne \q^2$.

\medskip\noindent

Theorem~\ref{thm:ninterval} is a consequence of the following two results.	
\begin{theorem}\label{thm:appendixintsecondfrontier}
		Let $\q = \Diamond (\rho_1 \land \nxt (\rho_2 \land \dots \land \nxt \rho_n)\in \Qint^{\sigma}$. A weakening frontier for $\q$ in $ \Qint^{\sigma}$ can be computed in polytime by dropping some atom from $\rho_{i}$ for some $i\in [n]$; here, if $i=1$ and $\rho_{i}$ is a singleton, then also drop $\top$ and take $\Diamond (\rho_2 \land \nxt(\rho_{3} \land \dots \land \nxt \rho_n)$.
\end{theorem}
\begin{proof}
	Let $\mathcal{F}$ denote the set of queries defined in Theorem~\ref{thm:appendixintsecondfrontier}. It is easy to see using Lemma~\ref{lem:containment} that $\q\models \q'$ and $\q\not\equiv \q'$ for all $\q'\in \mathcal{F}$.
	
	Assume now $\q\models \q'$ and $\q\not\equiv \q'$ for some $\q'\in \Qint^{\sigma}$.
	We show that there exists a query $\q''\in \mathcal{F}$ such that $\q''\models \q'$. Assume
	$$
	\q'= \Diamond (\rho_{1}' \wedge \nxt (\rho_{2}' \wedge \cdots \wedge\nxt \rho_{m}')).
	$$
	and assume that $g:[m] \rightarrow [n]$
	is a containment witness for $\q,\q'$.
	
	We distinguish the following cases.
	
	Case 1. $g$ is a surjective mapping onto $[n]$. Then $n=m$ and $g(i)=i$ for all $i\in [m]$. Hence $\rho_{i}'\subseteq \rho_{i}$ for all $i\in [m]$. Then, as $\q\not\equiv \q'$,  $\rho_{i}'\subsetneq \rho_{i}$ for some $i\in [m]$. Then  $\q''\models \q'$ for a $\q''\in \mathcal{F}$ obtained by dropping an atom from $\rho_{i}$ (observe that $\rho_{i}'\not=\top$ by the definition of interval queries, so $\rho_{i}$ is not a singleton).
	
	Case 2. $g$ is not onto $[n]$. Then either $i=1$ or $i=n$ are not in the range of $g$. If $i=1$, then we obtain $\q''\in \mathcal{F}$ with $\q''\models \q'$ by dropping an atom from $\rho_{1}$ (if $\rho_{1}$ is a singleton we also drop $\top$). If $i=n>1$,
	then we obtain $\q''\in \mathcal{F}$ with $\q''\models \q'$ by dropping an atom from $\rho_{n}$. 
\end{proof}	
\begin{theorem}\label{thm:appendixintfirstfrontier}
		Let $\q = \Diamond (\rho_1 \land \nxt (\rho_2 \land \dots \land \nxt \rho_n)\in \Qint^{\sigma}$. A strengthening frontier for $\q$ in $\Qint^{\sigma}$ can be computed by applying once to $\q$ one of the following operations\textup{:}
		\begin{enumerate}
			\item extend some $\rho_{i}$ in $\q$ by some $A\in\sigma\setminus \rho_{i}$\textup{;}
			
			\item insert $A \land \nxt^{n}$ before $\rho_{1}$, that is, form $\Diamond (A \land \nxt^{n}(\rho_1 \land \nxt (\rho_2 \land \dots \land \nxt \rho_n)$ for some $A\in \sigma$ and $n\geq 1$\textup{;}
			
			\item add $\nxt^{n}A$ at the end of $\q$, for some $A\in \sigma$ and $n\geq 1$.
		\end{enumerate}
	\end{theorem}	
\begin{proof}
Let $\fr$ be the set of queries obtained by a single application of one of these operations to $\q$. By Lemma~\ref{lem:containment}, $\q'\models \q$ and $\q\not\equiv \q'$ for all $\q'\in \fr$. Let $\q'\models \q$ and $\q\not\equiv \q'$, for some $\q'\in \Qint^{\sigma}$ of the form
$$
\q'= \Diamond (\rho_{1}' \wedge \nxt(\rho_{2}' \wedge \cdots \wedge \nxt \rho_{m}')).
$$
Take a containment witness $h \colon [n] \rightarrow [m]$ for $\q',\q$. If $h$ is surjective, then $n=m$ and $h(i)=i$ for all $i\in [n]$, and so $\rho_{i}\subseteq \rho_{i}'$. As $\q\not\equiv \q'$, $\rho_{i}\subsetneq \rho_{i}'$, for some $i\in [n]$. Then operation 1 gives a $\q''\in \fr$ with $\q'\models \q''$.

Suppose $h$ is not surjective. Then either $i=1$ is not in the range of $h$ or $i=m$ is not in the range of $h$. In the first case  $\q'\models \q''$ for a $\q''\in \fr$ given by operation 2. In the second case, $\q'\models \q''$ for a $\q''\in \fr$ given by operation 3. 
\end{proof}

\section{Proofs for Section~\ref{sec:complexity}}

\medskip
\noindent
{\bf Theorem~\ref{thm:mostspeciverif}.}
{\em For $\Qd$, the most specific separator verification problem coincides with the unique most specific separator verification problem and is \coNP-complete.}

\begin{proof}
	It remains to show the lower bound. We reduce SAT to the complement of most specific separator verification. Let $\varphi=\psi_1\wedge\ldots\wedge \psi_k$ be a
	propositional formula with $k$ clauses over $n$ variables
	$x_1,\ldots,x_n$. We assume without loss of generality that no
	variable occurs both positively and negatively in a clause. We use $2n$ atoms $A_1,\bar
	A_1,\ldots,A_n,\bar A_n$. Define the positive examples
	$E^+=\{\D_0,\D_0',\D_1,\ldots,\D_n\}$ by setting
	\begin{itemize}
		
		\item $\D_0=\{A_i(1),\bar A_i(1)\mid 1\leq i\leq n\}$,
		
		\item $\D_0'=\{A_i(2),\bar A_i(2)\mid 1\leq i\leq n\}$,
		
	\end{itemize}
	and including in $\D_i$, for $1\leq i\leq n$, the
	assertions:
	\begin{itemize}
		
		\item $A_i(1)$ and $\bar A_i(2)$, and
		
		\item $A_j(1),\bar A_j(1),A_j(2),\bar A_j(2)$ for all
		$j\neq i$.
		
	\end{itemize}
	We also add to each $\D_{i}$ the following, for $1\leq \ell\leq
	k$:
	\begin{itemize}
		
		\item $\bar A_j(2+\ell)$, if $x_j$ does not occur negatively in $\psi_{\ell}$, and
		
		\item $A_j(2+\ell)$, if $x_j$ does not occur positively in $\psi_{\ell}$.
		
	\end{itemize}
	
	Now consider the query
	$$
	\q= \bigwedge_{1\leq \ell \leq k}\Diamond\rho_{\ell}
	$$
	with $\rho_{\ell}$ containing:
	\begin{itemize}
		
		\item $\bar A_j$, if $x_j$ does not occur negatively in $\psi_{\ell}$, and
		
		\item $A_j$, if $x_j$ does not occur positively in $\psi_{\ell}$.
		
	\end{itemize}
	The negative examples are not relevant, we simply take $E^{-}=\{\D^{-}\}$ with $\D^{-}=\{A_{1}(0)\}$.
	
	Then clearly $\D,0\models \q$ for all $\D \in E^{+}$ and $\D^{-},0\not\models\q$. We show that $\q$ is equivalent to
	the unique most specific separator if and only if $\varphi$ is
	not satisfiable. For a variable assignment $\mathfrak a$ for $x_1,\ldots,x_n$,
	we denote with $\rho_\mathfrak a$ the set that contains, for all $i$, $A_i$, if
	$\mathfrak a(x_i)=1$ and $\bar A_i$, otherwise.
	
	Assume first that $\varphi$ is satisfiable. Take a variable
	assignment $\mathfrak a$ witnessing this. Then
	$\Diamond\rho_{\mathfrak a}$ separates $E$ but
	$\q\not\models\Diamond\rho_{\mathfrak a}$ and so $\q$ is not a
	most specific separating query for $E$. To show that
	$\q\not\models\Diamond\rho_{\mathfrak a}$
	take the data instance $\D$ containing
	for $1\leq \ell\leq
	k$:
	\begin{itemize}
		
		\item $\bar A_j(2+\ell)$, if $x_j$ does not occur negatively in $\psi_{\ell}$, and
		
		\item $A_j(2+\ell)$, if $x_j$ does not occur positively in $\psi_{\ell}$.
		
	\end{itemize}
	Then, by definition, $\D,0\models \q$. But $\D,0\not\models
	\Diamond\rho_{\mathfrak a}$ since $\mathfrak a$ is a satisfying variable assignment
	for $\varphi$.
	
	For the converse direction, assume that $\q'$ is a query in
	$\Qd$ that separates $E$ but that $\q\not\models \q'$. We show that
	$\varphi$ is satisfiable. By the positive examples $\D_{0}$
	and $\D_{0}'$ we may assume that $\q'$ is a conjunction 
	\[\q'=\Diamond\rho'_1\wedge \ldots \wedge \Diamond \rho'_\ell\] 
	with each $\rho'_j$ a set of atoms. It should be clear that
	each $\Diamond\rho'_j$ separates $E$ and there is at least one $j$ with $\q\not\models \Diamond\rho'_j$. Thus, we can assume that $\q'$ actually
	takes the form $\Diamond\rho$ for some set of atoms $\rho$.

	By
	construction of the positive examples $\D_i$, $1\leq i\leq n$, either $\rho\subseteq \rho_{\ell}$ for some $1\leq \ell \leq k$, or for every $i$ at most one of
	$A_i,\bar A_i$ is contained in $\rho$. As $\q\not\models \q'$,
	the first condition does not hold and the latter holds. Thus,
	$\rho\subseteq \rho_\mathfrak a$ for some variable
	assignment $\mathfrak a$. It remains to show that $\mathfrak a$ is a model for
	$\varphi$. But this follows directly from $\rho_{\mathfrak a}\not\subseteq \rho_{\ell}$ for any $\rho_{\ell}$, $1\leq \ell \leq k$.
\end{proof}

\medskip
\noindent
We do not know whether Theorem~\ref{thm:mostspeciverif} still holds if the number of positive examples is bounded.

\medskip\noindent
{\bf Theorem~\ref{thm:uniqueveri}.}
{\em Unique most specific/general separator verification in $\Qpnd$, $\Qpd$, $\Qint$ as well
as unique most general separator verification in $\Qd$ are \coNP-complete.}

\medskip
\noindent
\begin{proof}
	For the upper bound we show an \NP{} upper bound for the complement.
	For given $\q$ and $E$, check that either $\Q$ is not a most general/specific separator (this is in \PTime{} by Corollary~\ref{cor:most}) or there
	exists $\q'$ with $\q'\not\equiv \q$ which separates $E$ (which can be done in \NP{} by guessing $\q'$ and then doing the checks in \PTime{}).
	
	For the lower bounds we show an \NP{} lower bound for the complement.
	We use that given a propositional formula $\varphi$ and a
	satisfying assignment $\mathfrak a$ it is \NP{}-hard to
	decide whether there is satisfying assignment for $\varphi$
	different from $\mathfrak a$.
	
	Let $\varphi=\psi_1\wedge\ldots\wedge \psi_k$ a
	propositional formula with $k$ clauses over $n$ variables
	$x_1,\ldots,x_n$. Let $\mathfrak a_{0}$ be a satisfying assignment. We assume without loss of generality that no
	variable occurs both positively and negatively in a clause. We use $2n$ atoms $A_1,\bar
	A_1,\ldots,A_n,\bar A_n$. Define the positive examples
	$E^+=\{\D_0,\D_0',\D_1,\ldots,\D_n\}$ by setting
	\begin{itemize}
		
		\item $\D_0=\{A_i(1),\bar A_i(1)\mid 1\leq i\leq n\}$,
		
		\item $\D_0'=\{A_i(2),\bar A_i(2)\mid 1\leq i\leq n\}$,
		
	\end{itemize}
	and including in $\D_i$, for $1\leq i\leq n$, the
	assertions:
	\begin{itemize}
		
		\item $A_i(1)$ and $\bar A_i(2)$, and
		
		\item $A_j(1),\bar A_j(1),A_j(2),\bar A_j(2)$ for all
		$j\neq i$.
		
	\end{itemize}
	We next define
	$E^-=\{\D_1^1,\ldots,\D_k^1,\D^2_1,\ldots,\D^2_n\}$. In $\D_i^1$, $1\leq i\leq
	k$, we include
	\begin{itemize}
		
		\item $\bar A_j(1)$, if $x_j$ does not occur negatively in $\psi_i$, and
		
		\item $A_j(1)$, if $x_j$ does not occur positively in $\psi_i$.
		
	\end{itemize}
	In $\D_i^2$, $1\leq i\leq n$, we include
	\begin{itemize}
		
		\item $A_j(1),\bar A_j(1)$, for all $j\neq i$.
		
	\end{itemize}
	%
		%
		%
	%
	Set $E=(E^+,E^-)$. For a variable assignment $\mathfrak a$ for $x_1,\ldots,x_n$,
	we denote with $\rho_\mathfrak a$ the set that contains, for all $i$, $A_i$, if
	$\mathfrak a(x_i)=1$ and $\bar A_i$, otherwise.
	
	\begin{lemma}\label{lem:basic-reduction0}
		\begin{enumerate}[label=(\roman*)]
			
			\item For every model $\mathfrak a$ of
			  $\varphi$, $\Diamond \rho_\mathfrak a$
			separates $E$.
			
			\item If $\q\in\Qpnd$ separates $E$, then
			  $\q\equiv\Diamond\rho_\mathfrak a$ for
			some model $\mathfrak a$ of $\varphi$.
		\end{enumerate}
		%
		%
		%
	\end{lemma}
	\begin{proof}	
		Point~(i) is straightforward.
		
		For Point~(ii), take any separator \q for $E$. Since the positive
		example $\D_0$ satisfies $\max \D_0 = 1$, we can assume that \q has
		temporal depth at most 1. It cannot be of temporal depth 0 as then it would not separate.
		Thus, \q has temporal depth $1$ and is of shape
		$\q = \op \rho$ for some $\rho\subseteq \{A_1,\bar
		A_1,\ldots,A_n,\bar A_n\}$ and $\op\in\{\Diamond,\nxt\}$.
		Since~$\D_0'$ is a positive example, $\op$ cannot be $\nxt$ and $\q
		= \Diamond\rho$.
		By
		construction of the positive examples $\D_i$, $1\leq i\leq n$, for every $i$ at most one of
		$A_i,\bar A_i$ is contained in $\rho$. Since all $\D^2_i$ are
		negative examples, also at least one of $A_i,\bar A_i$ is
		contained in $\rho$. Thus, $\rho=\rho_\mathfrak a$ for some variable
		assignment $\mathfrak a$. It remains to show that
		$\mathfrak a$ is a model for
		$\varphi$. Let $\psi_i$ be a clause of $\varphi$ and assume that
		$\mathfrak a\not\models\psi_i$. Then $\mathfrak a(x_j)=1$ (and hence $A_j\in \rho$ if
		$x_j$ occurs negatively in $\psi_i$ and $\mathfrak a(x_j)=0$ (and hence
		$\bar A_j\in \rho$) if $x_j$ occurs positively. It is now
		readily checked that $\Dmc_i^1,0\models \Diamond\rho$, a contradiction.
	\end{proof}
	Observe that the $\Diamond\rho_{\mathfrak a}$ are pairwise
	incomparable w.r.t.\ containment. Hence, it follows from
	Lemma~\ref{lem:basic-reduction0} that $\Diamond\rho_{\mathfrak
	  a_{0}}$ is not the unique most specific/general separator in
	  $\Qpnd$ (equivalenty, $\Qpd$ or $\Qint$) iff there is an
	  assignment $\mathfrak a$ distinct from $\mathfrak a_{0}$ satisfying $\varphi$.
	
	For $\Qd$ the same argument applies, but only for unique most general separators.
\end{proof}

\medskip\noindent
{\bf Theorem~\ref{lem:consnew}.}
{\em Let $\mathcal{Q}\subseteq \Qpnd$ be any class  containing all $\Diamond\rho$ with $\rho$ a conjunction of atoms. Then
\begin{itemize}
	\item counting most specific/general $\mathcal{Q}$-separators is $\sharp P$-hard.
	\item the existence of unique most specific/general separator is US-hard.
\end{itemize}
even for $E=(E^{+},E^{-})$ and $\sigma=\sig(E)$ with
\begin{enumerate}
	\item $|E^{-}|\leq 1$; or
	
	\item $|E^{-}| \leq 1$, $|\sigma| = 2$ and
	$\mathcal{Q}=\Qint$ or $\mathcal{Q}\supseteq \Qpd$; or
	
	\item $|E^{+}| \leq 4$, $|\sigma| = 3$, and $\mathcal{Q}\supseteq \Qpd$.
\end{enumerate}
This result also holds for most general separator case for $\Qd$ and $|E^{-}|=1$ and $|\sigma|=2$.
}

\medskip\noindent
\begin{proof}
We start with the proof of Point~1. For the sake of readability, we
	provide the proof first for an unbounded number of negative examples and
	argue afterwards that a single negative examples suffices.
	
	We reduce from SAT in essentially the same way as in the lower bound proof for Theorem~\ref{thm:uniqueveri} above. Let $\varphi=\psi_1\wedge\ldots\wedge \psi_k$ a
	propositional formula with $k$ clauses over $n$ variables
	$x_1,\ldots,x_n$. We assume without loss of generality that no
	variable occurs both positively and negatively in a clause. We use $2n$ atoms $A_1,\bar
	A_1,\ldots,A_n,\bar A_n$. Define the examples
	$E=(E^{+},E^{-})$ in exactly the same way as in the proof of Theorem~\ref{thm:uniqueveri}. Then 
	Lemma~\ref{lem:basic-reduction0} holds. Point~1 for an unbounded number of negative examples follows directly.
	
	To finish the proof for $|E^{-}|=1$, it remains to join the negative examples into a single one.  Recall
	that, due to the structure of the positive examples, any separator for $E$ has temporal depth at
	most~$1$. Based on this observation, it can be verified that
	Lemma~\ref{lem:basic-reduction0} continues to hold when $E^-$ is
	replaced by $\{\D^-\}$ for the data
	instance
	\[\D^-=\emptyset\D^1_1\emptyset^2\cdots\emptyset^2 \D^1_k\emptyset^2
	\D^2_1\emptyset^2\cdots \emptyset^2\D^2_n.\]

	\bigskip
	We next provide the proof of Point~2. 
	We start by considering $\Qint$. Again, we
	provide the proof first for an unbounded number of negative examples and
	argue afterwards that a single negative examples suffices.
	
	\paragraph{\Qint} We reduce from SAT. Let $\varphi=\psi_1\wedge\ldots\wedge \psi_k$ a
	propositional formula with $k$ clauses over $n$ variables
	$x_1,\ldots,x_{n}$. We use three atoms $M,A,\bar
	A$. Define the positive examples
	$E^+=\{\D_0,\D_1,\ldots,\D_{n}\}$ by including in $\D_i$, for $1\leq
	i\leq n$ the
	assertions:
	\begin{itemize}
		
		\item $M(1)$ and $M(n+1)$, and
		
		\item $A(i)$ and $\bar A(n+i)$, and
		
		\item $A(j),\bar A(n+j)$ for all $j\neq i$.
		
	\end{itemize}
	In $\D_0$, we include
	\begin{itemize}
		
		\item $M(1)$, and
		
		\item $A(i),\bar A(i)$, for all $1\leq i\leq n$.
		
	\end{itemize}
	We next define negative examples
	$E^-=\{\D^0,\D_1^1,\ldots,\D_k^1,\D^2,\D^3_1,\ldots,\D^3_n\}$.
	In $\D^0$, we include
	\begin{itemize}
		
		\item $M(i),A(i),\bar A(i)$, for all $1\leq i \leq n-1$.
		
	\end{itemize}
	In $\D_i^1$, $1\leq i\leq
	k$, we include
	\begin{itemize}
		
		\item $M(1)$, and
		
		\item $\bar A(j)$, if $x_j$ does not occur negatively in $\psi_i$, and
		
		\item $A(j)$, if $x_j$ does not occur positively in $\psi_i$.
		
	\end{itemize}
	In $\D^2$, we include
	\begin{itemize}
		
	  \item $A(j),\bar A(j)$, for all $1\leq j\leq n$.
		
	\end{itemize}
	In $\D^3_i$, $1\leq i\leq n$, we include
	\begin{itemize}
		
		\item $M(1)$,
		
		\item $A(j),\bar A(j)$, for all $i\neq j$.
		
	\end{itemize}
	Set $E=(E^+,E^-)$. For a variable assignment $\mathfrak a$ for $x_1,\ldots,x_n$,
	we denote with $\q_\mathfrak a$ the query
	\[\q_\mathfrak a = \Diamond (\rho^1_\mathfrak a\wedge (\nxt
	  \rho^2_\mathfrak a\wedge (\ldots \wedge
	\nxt\rho_\mathfrak a^n)\ldots ))\]
	where $\rho_\mathfrak a^1=\{M,A\}$ if $\mathfrak a(x_1)=1$ and
	$\rho_\mathfrak a^1=\{M,\bar A\}$ otherwise, and
	for all $i>1$, $\rho_\mathfrak a^i=\{A\}$ if $\mathfrak
	a(x_i)=1$ and $\rho_\mathfrak a^i=\{\bar
	A\}$, otherwise.
	\begin{lemma} \label{lem:basic-reduction2}
		\begin{enumerate}[label=(\roman*)]
			
			\item For every model $\mathfrak a$ of
			  $\varphi$, $\q_\mathfrak a$
			separates $E$.
			
			\item If $\q\in\Qint$ separates $E$, then
			  $\q\equiv\q_\mathfrak a$ for
			some model $\mathfrak a$ of $\varphi$.
			
			\item Every $\Qint$-separator for $E$
			is both a most general and most specific
			separator.
			
		\end{enumerate}
		%
		%
		%
		%
		%
	\end{lemma}

	\begin{proof}
		Point~(iii) is immediate from Points~(i)
		and~(ii) and the fact that all $\q_\mathfrak a$ are
		pairwise incomparable w.r.t.\ containment. We hence concentrate on Points~(i)
		and~(ii).
		
		Point~(i) is straightforward.
		
		For Point~(ii), take any separator
		\[\q = \Diamond (\rho^1\wedge (\nxt \rho^2\wedge (\ldots \wedge
		\nxt\rho^m)\ldots ))\]
		for $E$, for some $m>0$ and sets $\rho^i$. Since $\D^0$ is a
		negative example for \q, we have $m\geq n$.
		Since $\D_0$ is a positive example for $\q$, we have
		that $m\leq n$,
		hence $m=n$. Then observe that some $\rho^i$
		has to contain $M$, since otherwise $\Dmc^2,0\models \q$. By $\D_0$
		again, it can only be $\rho^1$.
		
		By construction of the positive examples, for every $i$ at most one of
		$A,\bar A$ is contained in $\rho^i$. Since all $\D^3_i$ are
		negative examples, each $\rho_i$ contains also at least one of
		$A,\bar A$. Thus, $\q=\q_\mathfrak a$ for some variable
		assignment $\mathfrak a$.
		
		It remains to show that $\mathfrak a$ is a model for $\varphi$. Let $\psi_i$
		be a clause of $\varphi$ and assume that $v\not\models\psi_i$ for
		some $i$. Then $\mathfrak a(x_j)=1$ (and hence $A\in \rho_j)$ if $x_j$ occurs negatively in
		$\psi_i$ and $\mathfrak a(x_j)=0$ (and hence $\bar A\in \rho_j$) if $x_j$
		occurs positively. It is now readily checked that $\Dmc_i^1,0\models
		\q_\mathfrak a$, a contradiction.
	\end{proof}
	
	It remains to join the negative examples into a single one. Recall
	that, because of $\D_0\in E^+$, any separator for $E$ has temporal depth at
	most $n$. Based on this observation, it can be verified that
	Lemma~\ref{lem:basic-reduction2} continues to hold when $E^-$ is
	replaced by $\{\D^-\}$ for the data
	instance
	\[\D^-=\emptyset\D^0\emptyset^{n} \D^1_1\emptyset^n\cdots\emptyset^n\D^1_k\emptyset^n
	\D^2\emptyset^n\D^3_1\emptyset^n\cdots \emptyset^n\D^3_n.\]

	\bigskip
	We next assume that $\Qpd \subseteq \mathcal{Q}\subseteq \Qpnd$.
	
	We employ a reduction of 3SAT to a longest common subsequence problem, which
	is defined as follows: Given a set $S$ of sequences and a number $m$,
	decide whether the sequences in $S$ have a common subsequence of
	length $m$. We use the following result~\cite[Proposition 6]{DBLP:conf/cpm/BlinBJTV12}.
	Given a formula $\varphi=\psi_{1}\wedge \cdots \wedge \psi_{k}$ with
	$k$ clauses over $n$ variables $x_{1},\ldots,x_{n}$, one can construct
	in polynomial time a set $S$ of sequences over the alphabet $\{H,T\}$ such that
	\begin{itemize}
		\item $S$ contains $k+2$ sequences;
		\item there is a poly-time computable bijection $f$
		  from the set $L$ of common subsequences of $S$ of
		  length exactly $11n$ onto the set $V$ of variable
		  assignments $\mathfrak a$ satisfying $\varphi$;
		\item no common subsequence of the sequences in $S$ of length $>11n$ exists.
	\end{itemize}
	Thus, given $\varphi$ as above, we use $S$ to construct $E=(E^{+},E^{-})$. For any sequence $s_{1}\cdots s_{\ell}\in S$ we include in $E^{+}$ the data instance $\emptyset\{s_{1}\}\cdots\{s_{\ell}\}$. We also include in $E^{+}$ the data instances $\D_{1}^{+}=\emptyset\{H,T\}^{11n}$ and $\D_{2}^{+}=\emptyset(\emptyset^{11n}\{H,T\})^{11n}$
	and include in $E^{-}$ the data instances $\D_{1}^{-}=\emptyset\{H,T\}^{11n-1}$ and
	$\D_{2}^{-}=\emptyset(\emptyset^{11n}\{H,T\})^{11n-1}$.
	
	\medskip\noindent\textit{Claim~1.} Any query $\q\in
	\mathcal{Q}$ separating $E$ is a $\Diamond$-path
	sequence query and the sequence defined by $\q$ is a common
	subsequence of the sequences in $S$.
	
	\medskip\noindent\textit{Proof of Claim~1.} To prove Claim 1 suppose $\q\in \mathcal{Q}$ of the form
	\begin{equation}
		\q = \rho_0 \land \op_{1} (\rho_1 \land \op_{2} (\rho_2 \land \dots \land \op_{\ell}\rho_{\ell}) ),
	\end{equation}
	separates $E$. By definition of separation it is clear that $\rho_{0}=\top$ and $|\rho_{i}|\leq 1$. From $\D_{1}^{+}\in E^{+}$ it follows that $\ell\leq 11n$
	and from $\D_{1}^{-}\in E^{-}$ it follows that $\ell\geq 11n$. Hence $\ell=11n$.
	
	It follows from $\D_{2}^{+}\in E^{+}$ and $\ell= 11n$ that
	$\op_{i}=\Diamond$ for all $i\leq \ell$. We obtain from
	$\D_{2}\in E^{-}$ that $|\rho_{i}|\geq 1$ for all $i\geq 1$.
	Thus, $\q$ is a $\Diamond$-path sequence query of length
	$11n$.  Since we included one positive example for every
	sequence in $S$, the sequence corresponding to $\q$ is a common
	subsequence of the sequences in $S$. This finishes the proof
	of Claim~1.
	
	\medskip
	
	The following claim follows directly from the definition of $E$.
	
	\medskip\noindent\textit{Claim~2.} For every $s_{1}\ldots s_{\ell}\in S$, the corresponding $\Diamond$-path query $\Diamond(s_{1} \land \Diamond (s_{2} \cdots \Diamond s_{\ell}))$ separates $E$.
	
	\medskip
	
	It follows from Claims~1 and~2 and the fact that all sequence
	$\Diamond$-path queries of the same temporal depth are pairwise
	incomparable w.r.t.\ containment, that there is a 
	poly-time computable bijection $f$ from the set $L'$ of queries in
	$\mathcal{Q}$ separating $E$ onto the set $V$ of
	satisfying assignments for $\varphi$. 

\bigskip
We now come to the proof of Point~3. 
We use an existing reduction showing \NP-hardness of the
\emph{consistent subsequence problem}, defined as follows. Given two
sets $P,N$ of sequences, decide whether there is a sequence $s$
\emph{consistent with $P,N$}, that is, $s$ is a
	subsequence of each sequence in $P$, but of no sequence in $N$.
	
	We use the following
	result~\cite[Theorem~2.1]{DBLP:journals/tcs/Fraser96}. Given a 3CNF
	formula $\varphi=\psi_1\wedge\ldots\wedge \psi_k$ with $k$ clauses
	over $n$ variables, one can construct in polynomial time two sets
	$P,N$ of sequences over an alphabet $\sigma=\{\#,0,1\}$ such that
	\begin{itemize}
		
		\item $P$ contains two sequences;
		
		\item there is a poly-time computable bijection $f$ from the set
		$L$ of sequences that are consistent with $P,N$ of length $6k$ onto the
		set of models of $\varphi$;
		
		\item any sequence that is consistent with $P,N$ has length exactly $6k$.
		
	\end{itemize}
	Thus, given $\varphi$ as above, we use $P,N$ to construct
	$E=(E^{+},E^{-})$. For any sequence $s_{1}\cdots s_{\ell}\in P$ we
	include in $E^{+}$ the data instance
	$\emptyset\{s_{1}\}\cdots\{s_{\ell}\}$. Similarly, for any sequence
	$s_1\cdots s_\ell\in N$, we include in $E^-$ the data instance
	$\emptyset\{s_1\}\cdots\{s_\ell\}$. We also include in $E^{+}$
	the data instances
	$\D_{1}^{+}=\emptyset\{0,1,\#\}^{6k}$ and
	$\D_{2}^{+}=\emptyset(\emptyset^{6k}\{0,1,\#\})^{6k}$,
	and include in $E^{-}$ the data instances
	$\D_{1}^{-}=\emptyset(\{0,1,\#\})^{6k-1}$ and
	$\D_{2}^{-}=\emptyset(\emptyset^{6k}\{0,1,\#\})^{6k-1}$.
	
	Assume now that a set $\mathcal{Q}$ with
	$\mathcal Q_p[\Diamond]\subseteq \mathcal Q\subseteq \mathcal
	Q_p[\nxt\Diamond]$ is given.
	
	\medskip\noindent\textit{Claim~1.} Any query $\q\in \mathcal{Q}$
	separating $E$ is a $\Diamond$-path sequence query and
	the sequence defined by $\q$ is consistent with $P,N$.
	
	\medskip\noindent\textit{Proof of Claim~1.} Suppose $\q\in
	\mathcal{Q}$ of the form
	\[\q = \rho_0 \land \op_{1} (\rho_1 \land \op_{2} (\rho_2 \land
	\dots \land \op_{\ell}\rho_{\ell}) )\]
	separates $E$. By definition of separation it is clear that
	$\rho_{0}=\top$ and $|\rho_{i}|\leq 1$. From $\D_{1}^{+}\in E^{+}$
	it follows that $\ell\leq 6k$ and from $\D_{1}^{-}\in E^{-}$ it
	follows that $\ell\geq 6k$. Hence $\ell=6k$. From $\D_2^-\in E^-$ it
	also follows that $\rho_i\neq\emptyset$ for all $i>1$, hence $\q$ is
	a $\nxt\Diamond$-path sequence query.
	
	It follows from $\D_{2}^{+}\in E^{+}$ and $\ell=6k$ that
	$\op_{i}=\Diamond$ for all $i\leq \ell$. We have thus shown that
	$\q$ is a $\Diamond$-path sequence query of length $6k$. Finally, it
	is easy to see that the corresponding sequence is consistent with
	$P,N$. This finishes the proof of Claim~1.

	\medskip The following claim follows directly from the definition of
	$E$.
	
	\medskip\noindent\textit{Claim~2.} For every sequence consistent
	with $P,N$, the corresponding $\Diamond$-path query separates $E$.
	
	\medskip It follows from Claims~1 and 2 to that we have a poly-time
	computable bijection $f$ from the set $L'$ of queries in
	$\mathcal{Q}$ separating $E$ onto the set of models of $\varphi$. Moreover, the queries in $L'$
	are all $\Diamond$-path sequence queries of the same length and
	therefore incomparable w.r.t.~$\models$.
	%
		%

\bigskip

We finally come to $\Qd$. Observe that our claim follow directly from Point~2 for $\Qpd$:
as we have only one negative example, any separating query in $\Qd$ contains a conjunct in $\Qpd$ that separates. Hence the most general separating queries
in $\Qd$ and $\Qpd$ coincide (up to logical equivalence). So the $\sharp \PTime$-lower bound established for $\Qpd$ for counting most general separating queries and the US-lower bound for the existence of a unique most general separator also hold for $\Qd$.
\end{proof}

\medskip\noindent
{\bf Theorem~\ref{lem:consnew+}.}
	{\em Let $\Q=\Qd$ or $\Q \subseteq \Qpnd$ be any class containing  all $\Diamond\rho$ with a conjunction of atoms $\rho$. Then counting most general $\mathcal{Q}$-sepa\-ra\-tors is $\sharp \PTime$-hard even for example sets $E = (E^+,E^-)$ with $|E^{+}|=2$ and $|E^{-}|=1$.}
	
\medskip\noindent
\begin{proof}
	By reduction from counting the number of satisfying assignments
	for a monotone propositional formula (that is, a formula in conjunctive normal form without negation).
	We re-use data instances introduced before.
	
	Take a monotone Boolean formula $\varphi=\psi_1\wedge \dots \wedge \psi_k$ with clauses $\psi_i$ over variables $x_1, \dots, x_n$. 
	We construct $E=(E^{+},E^{-})$ such that there is a bijection between the satisfying assignments for $\varphi$ and the most general separators for $E$ in $\Qpnd$ (even queries of the form $\Diamond\rho$). Define the positive examples $E^+ = \{\D_0, \D_0'\}$ with $2n$ atoms $A_1,\bar A_1, \dots, A_n,\bar A_n$ by taking 
	\begin{align*}
		& \D_0 = \{A_i(1),\bar A_i(1)\mid 1\leq i\leq n\},\\
		& \D_0'=\{A_i(2), \bar A_i(2)\mid 1\leq i\leq n\}
	\end{align*}
	Let $E^{-}=\{\D^{-}\}$ with 
	
\[\D^-=\emptyset\D^1_1\emptyset^2\cdots\emptyset^2 \D^1_k\emptyset^2
\D^2_1\emptyset^2\cdots \emptyset^2\D^2_n.\]
	where $\D_i^1$, $1\leq i\leq k$, comprises all $\bar A_j(1)$ and $A_j(1)$ if $x_j$ does not occur positively in $\psi_i$, and 
	$\D_i^2 = \{A_j(1),\bar A_j(1) \mid j\neq i\}$.
	
	For a variable assignment $\mathfrak a$ for $x_1,\ldots,x_n$,
	we denote with $\rho_\mathfrak a$ the set that contains, for all $i$, $A_i$, if
	$\mathfrak a(x_i)=1$ and $\bar A_i$, otherwise. Now one can easily show:
	\begin{enumerate}[label=(\roman*)]
		
		\item For every model $\mathfrak a$ of $\varphi$,
		  $\Diamond \rho_\mathfrak a$
		separates $E$.
		
		\item If $\q\in\Qpnd$ separates $E$, then
		  $\q\equiv\Diamond(\rho_\mathfrak a \cup \rho)$ for
		some model $\mathfrak a$ of $\varphi$ and some set $\rho$ of atoms of the form $\bar A_{i}$.		
	\end{enumerate}
	The claim for $\Q\subseteq \Qpnd$ follows. The claim for $\Qd$ also follows since we have only one negative example and so any most general separator is in $\Qd$.
\end{proof}
We show the properties of $\Qint$ that are stated in Section~\ref{sec:complexity}.
\begin{theorem}
	(1) For $\Qint$, counting most specific separators is in $\PTime$ if $|E^{+}|$ is bounded.
	
	(2) For example data instances $E=(E^{+},E^{-})$ that contain sequences, counting most specific sequence queries in $\Qint$ that separate $E$ is in \PTime.
\end{theorem}

\begin{proof} We extend the proof of Theorem~\ref{thm:qinappendix}.
	 
	(1) Assume $E^{+}=\{\D_{1},\ldots,\D_{n}\}$. Let $\Xi$ be the set of queries defined in the proof of Theorem~\ref{thm:qinappendix}.
	Then the most specific queries in $\Xi$ that separate $E$ coincide with the most specific queries in $\Qint$ that separate $E$, and so can be counted in polytime.
		
	(2) $E^{+}=\{\D_{1},\ldots,\D_{n}\}$ be a set of sequences. Assume 
	$$
	\D_{1}= \{A_{0}\}\{A_{1}\}\cdots \{A_{m}\}
	$$
	Then any possible separating sequence query in $\Qint$ takes the form
	$$
	\Diamond(A_{k_{1}} \wedge \nxt(\cdots \nxt A_{k_{2}}))
	$$
	with $0<k_{1}\leq k_{2}\leq m$. Clearly the most specific queries of this form that separate $E$ can be counted in $\PTime$.
\end{proof}


\section{Proofs for Section~\ref{algs}}

Let $t_\D(n) = \{ A \mid A(n) \in \D\}$. It will be convenient to assume that $t_{\mathcal E_j}(\infty) = \sigma$, for $E_j \in E^-$, where $\sigma$ is the set of symbols occurring in $E^+$.

\medskip
\noindent
\textbf{Lemma~\ref{lem:sep-bounded-diamond-prod}.}
\emph{$\sep(E,\Q) \ne \emptyset$ iff $\mathfrak P \otimes \mathfrak N$ contains a separating path $\pi$, with $\q_\pi$ separating $E$.}
\begin{proof}
$(\Rightarrow)$ Suppose $\q$ of the form $\rho_0 \land \Diamond (\rho_1 \land ( \dots \Diamond \rho_\ell) \dots )$ separates $E$. If $\ell = 0$, then $\rho_0 \subseteq \avec{l}(\bar{0})$ and $\rho_0 \not \subseteq t_{\mathcal E_j}(0)$ for all $1 \leq j \leq c_-$. It follows that $\avec{l}(\bar 0) \neq \emptyset$ and $(\bar{0}, \bar \infty)$ is in $\mathfrak P \times \mathfrak N$, which completes the proof. Suppose $\ell > 0$. Further, we assume that $\rho_0 \land \Diamond (\rho_1 \land ( \dots \Diamond \rho_{\ell-1}) \dots )$ does not separate $E$, which implies $\rho_\ell \neq \emptyset$. We define $\avec{n}^0 = \bar{0}$. Suppose $\avec{n}^i = (n_1, \dots, n_{c_+})$ has been defined, we define $\avec{n}^{i+1} = (n_1', \dots, n_{c_+}')$ where $n_j' = \min \{ n \mid n > n_j, \rho_{i+1} \subseteq t_{\mathcal D_j}(n) \}$, for $0 \leq i < \ell$. It follows that $\avec{l}(\avec{n}^\ell) \neq \emptyset$. Then we take $\avec{m}^0$ to be the root of $\mathfrak N$ and set $\avec{m}^{i+1}$ to be $\avec{m}'$ such that $(\avec{n}^i, \avec{m}^i) \to (\avec{n}^{i+1}, \avec{m}')$. We claim that the path $(\avec{n}^0, \avec{m}^0), \dots, (\avec{n}^\ell, \avec{m}^\ell)$ is such that $\avec{m}^\ell = \bar \infty$. To prove the claim, we introduce some notation. For a data instance $\mathcal E$ and $t \geq 0$, denote by $\mathcal E^{\leq t}$ the data instance obtained from $\mathcal E$ by removing all assertions $A(n)$ with $n>t$. Set $m_j^i = \min \{ t \mid t \geq i \text{ and }\mathcal E_j^{\leq t} \models \rho_0 \land \Diamond (\rho_1 \land ( \dots \Diamond \rho_i) \dots )\}$ or $m_j^i = \infty$ if $t$ as above does not exist, for $0 \leq i \leq c_-$ and $1 \leq j \leq l$. Let $\avec{m}^i_- = (m_1^i, \dots, m_{c_-}^i)$. It follows that $\avec{m}^\ell_- = \bar \infty$. It is easily shown by induction on $i$ that $\avec{m}^i_- = \avec{m}^i$, which completes the proof.

$(\Leftarrow)$ Let $(\avec{n}^0, \avec{m}^0), \dots, (\avec{n}^\ell, \avec{m}^\ell)$ be a path with $\avec{n}^0 = \bar{0}$, $\avec{l}(\avec{n}^\ell) \neq \emptyset$ and $\avec{m}^\ell = \bar \infty$. We will show that $\q = \rho_0 \land \Diamond (\rho_1 \land ( \dots \Diamond \rho_\ell) \dots )$ with $\rho_i = \avec{l}(\avec{n}^i)$ separates $E$. Indeed, clearly $\D \models \q$ for each $\D \in E^+$. By induction on $i$, we show that $\avec{m}^i_- = \avec{m}^i$. It follows that $\mathcal E \not \models \q$ for all $\mathcal E \in E^-$ which completes the proof.
\end{proof}

\noindent
\textbf{Lemma~\ref{th:bounded-subs}.}
\emph{If $\q$ of the form~\eqref{dnpath} separates $E$, then there is a separating path of the form~\eqref{eq:prod-path} with $\rho_i \subseteq \avec{l}(\avec{n}_i)$, for all $i \le n$.}
\begin{proof}
Straightforwardly follows from the construction in $(\Rightarrow)$ above.
\end{proof}

\noindent
\textbf{Lemma~\ref{th:strongest-crit}.}
\emph{A $\Qpd$-query $\q$ is a unique most specific $\Qpd$-separator for $E$ iff there is a separating path $\pi$ such that $\q = \q_\pi$ and $\q_\pi \models \q_\nu$, for every separating path $\nu$.}
\begin{proof}
$(\Rightarrow)$ Let $\q = \rho_0 \land \Diamond (\rho_1 \land ( \dots \Diamond \rho_\ell) \dots )$ be a unique most specific separator. Consider the path $\pi = (\avec{n}^0, \avec{m}^0), \dots, (\avec{n}^\ell, \avec{m}^\ell)$ constructed in the proof of Lemma~\ref{lem:sep-bounded-diamond-prod} $(\Rightarrow)$. Clearly, $\pi$ is a separating path. Because $\q$ is a most specific separator and $\rho_i \subseteq \avec{l}(\avec{n}^i)$, it follows that $\avec{l}(\avec{n}^i) = \rho_i$, and so $\q = \q_\pi$. Let $\nu$ be an arbitrary separating path in $\mathfrak P \times \mathfrak N$. From the proof of Lemma~\ref{lem:sep-bounded-diamond-prod} $(\Leftarrow)$ we obtain that $\q_\nu$ is a separator. It follows that $\q_\pi \models \q_\nu$.

$(\Leftarrow)$ Let $\pi$ be a separating path such that $\q_\pi \models \q_\nu$ for every separating path $\nu$ in $\mathfrak P \otimes \mathfrak N$. Let $\q = \rho_0' \land \Diamond (\rho_1' \land ( \dots \Diamond \rho'_{\ell'}) \dots )$ be a separator (for $E$), we only need to show that $\q_\pi \models \q'$. Consider the separating path $\nu' = (\avec{n}^0, \avec{m}^0), \dots, (\avec{n}^{\ell'}, \avec{m}^{\ell'})$ constructed in the proof of Lemma~\ref{lem:sep-bounded-diamond-prod} $(\Rightarrow)$ for $\q'$. Clearly, $\rho_i' \subseteq \avec{l}(\avec{n}^i)$. Since we have $\q_\pi \models \q_{\nu'}$, we also obtain $\q_\pi \models \q'$.
\end{proof}

We now explain how checking the existence of a separating path $\pi$ such that $\q_\pi \models \q_\nu$ for every separating path $\nu$ in $\mathfrak P \otimes \mathfrak N$ can be checked in \PTime{}. We compute the required separator $\q_\pi$ by analysing the queries associated with the paths leading to each node $(\avec{n}, \avec{m})$ of $\mathfrak P \otimes \mathfrak N$ from its root $(\bar 0, \avec{m}_0)$. We mark inductively each $(\avec{n}, \avec{m})$ either by a query $\q_\pi$, for some path $\pi$, or by $\varepsilon$.
The root is marked by $\q_\pi$ for $\pi = (\bar 0, \avec{m}_0)$. Now, assume that all immediate predecessors of $(\avec{n}, \avec{m})$ have already been marked. Let $\q_{\pi_i}$, $i = 1,\dots,k$, be all of their marks different from $\varepsilon$. Consider the extended paths $\pi'_i = \pi_i, (\avec{n}, \avec{m})$ and the set $\Pi$ of all paths to $(\avec{n}, \avec{m})$. We check whether there is a containment witness for $\q_{\pi'_1}$ and every $\q_\pi$, $\pi \in \Pi$, in which case we mark $(\avec{n}, \avec{m})$ by $\q_{\pi'_1}$. Otherwise, we do the same for $\q_{\pi'_2}$, and so on. If $(\avec{n}, \avec{m})$ has no mark  after $k$ iterations, we mark it by $\varepsilon$. Although $\Pi$ can be of exponential size, we show that marking each node can be done in \PTime{}. Indeed, consider the following problem: given $\q_{\pi'_1}$ and $(\avec{n}, \avec{m})$, decide if there exists $\pi \in \Pi$ such that there is no containment witness for $\q_{\pi'_1}$ and $\q_\pi$. This problem is in \NL{} because such $\pi$, if exists, can be computed by a non-deterministic algorithm guessing $\pi$ of the form~\eqref{eq:prod-path} node-by-node. At each step $i \in \{0, \dots, n\}$ we only need to remember the minimal $h(i)$ (in the set $\{0, \dots, n'-1\}$ for the length $n'$ of $\pi_1'$) where $h$ is a containment witness for $\q_{\pi'_1}$ and $\q_{(\avec{n}_0, \avec{m}_0),\dots,(\avec{n}_i, \avec{m}_i)}$. If at some step it is impossible to extend $h(i)$ to $h(i+1)$ and $(\avec{n}, \avec{m})$ is reachable from $(\avec{n}_{i+1}, \avec{m}_{i+1})$, we have obtained $\pi \in \Pi$ such that there is no containment witness for $\q_{\pi'_1}$ and $\q_\pi$. Because $\NL{}= \textsc{co}\NL{}$ and the latter problem is in $\NL{}$, we conclude that checking whether there is a containment witness for $\q_{\pi'_1}$ and $\q_\pi$ for every $\pi \in \Pi$ is in $\NL$.

Then we consider all nodes $(\avec{n}_i, \bar \infty)$, $i = 1,\dots,l$, with $\avec{l}(\avec{n}_i) \ne \emptyset$ that are marked by some $\q_i \ne \varepsilon$. If there are no such, there is no unique most specific separator. Otherwise, take the set $\Pi$ of all paths leading to the $(\avec{n}_i, \bar \infty)$. We check whether there is a containment witness for $\q_1$ and every $\q_\pi$, $\pi \in \Pi$, in which case $\q_1$ is returned as a unique most specific separator. Otherwise, we do the same for $\q_{2}$, and so on.

\medskip
\noindent
\textbf{Theorem~\ref{thm:algorithms}.}
{\em Let $E = (E^+,E^-)$, $\sigma = \sig(E)$, $\tp$/$\tm$ is the maximum timestamp in $E^+$/$E^-$, $\cp = |E^+|$, $\cm = |E^-|$, and $\Q \in\{\Qpd, \Qpnd\}$.
The following can be done in time $O(\tpcp \tmcm)$\textup{:}}
\begin{itemize}\itemsep=0pt
\item[$(a)$] \emph{deciding whether $\sep(E,\Q)\ne \emptyset$\textup{;}}

\item[$(b)$] \emph{computing a longest/shortest separator in $\sep(E,\Q)$\textup{;}}

\item[$(c)$] \emph{deciding the existence of a unique most specific $\Q$-separator and a unique most general $\Qpd$-separator, and constructing such a separator.}
\end{itemize}
\emph{For bounded $\cp$ and $\cm$, problem $(a)$ is in $\NL$.}
\begin{proof}
We first prove $(c)$ for unique most general $\Qpd$-separators.
Computing unique most general separators requires a different type of graph. One might think that inverting $\q_\pi \models \q_\nu$ in Lemma~\ref{th:strongest-crit} would give a criterion for $\q_\pi$ being such a separator. To show that this is not so, take $E^+ = \{ \{A(1), B(1)\} \}$ and $E^- = \{ \{C(0)\} \}$. In this case, $\mathfrak P \otimes \mathfrak N$ looks as follows:\\
\centerline{
\begin{tikzpicture}[->,thick,node distance=1.8cm, transform shape,scale=0.9]\footnotesize
\node[rectangle, label=left:{$\emptyset$}] (s00) {\scriptsize$((0), (0))$};
\node[rectangle, right of = s00, label=right:{\scriptsize$\{A,B\}$}] (s11) {\scriptsize$((1),\vec{\infty})$} edge[above, <-] (s00);

\end{tikzpicture}
}\\[-2pt]
It contains only one separating path. However, there is no unique most general separator. The issue is that the labels of the  paths $\pi$ leading to $(\avec{n}, \vec{\infty})$ represent some of the separators, including most specific ones, but not all separators. In our example, both $\Diamond A$ and $\Diamond B$ are separators.

To this end, we consider a graph
$\mathfrak O$ using the nodes $-\bar{1}$ and $\avec{p}$, where $\avec{p} = (n_1, \dots, n_{c_+}, m_1, \dots, m_{c_-})$, $n_i \in [0, \max \D_i]$, and $m_i \in [0, \max \D_i] \cup \{\infty\}$. We define $\avec{l}(\avec{p}) = \{A \in \sigma \mid A(n_i) \in \D_i \text{ for all }1 \leq i \leq c_+ \text{ and }A(m_i) \in \D_i \text{ for all }1 \leq i \leq c_- \text{ such that }m_i \neq \infty\}$.
The edges of $\mathfrak O$ are labelled with $\rho \subseteq \sigma$
We treat both $\avec{l}(\avec{p})$ and $\rho$ as propositional (atemporal) queries and use the entailment relation for them, e.g., $\avec{l}(\avec{p}) \models \rho$.
We will slightly abuse the notation and assume $\avec{p} = (\avec{n}, \avec{m})$ if $\avec{p} = (n_1, \dots, n_{c_+}, m_1, \dots, m_{c_-})$, $\avec{n} = (n_1, \dots, n_{c_+})$ and $\avec{m} = (m_1, \dots, m_{c_-})$.
Define $\nabla_{\avec{p}, \avec{p}'}$ to be the set of $\avec{p}''$ such that $\avec{p}_i < \avec{p}''_i \leq \avec{p}_i'$ for all $i$ and $\avec{p}''_i < \avec{p}_i'$ for some $i$ such that $\avec{p}_i' \neq \infty$. To define the edges, we first let $-\bar{1} \to_\rho (\avec{n}, \avec{m})$ if $\avec{n}_i = 0$, $\avec{m}_i \in \{ 0, \infty\}$ for all $i$, $\rho \subseteq \avec{l}(\avec{p})$, and $\rho \not \subseteq \avec{l}(\avec{p}')$ for all $\avec{p}' \in \nabla_{\avec{0}, \avec{p}}$. Then, similarly, we let $\avec{p} \to_\rho \avec{p}'$ if $\rho \subseteq \avec{l}(\avec{p}')$, and $\rho \not \subseteq \avec{l}(\avec{p}'')$ for all $\avec{p}'' \in \nabla_{\avec{p}, \avec{p}'}$.
\begin{example}
For $E = (\emptyset, E^-)$ with $E^- = \{\D_1, \D_2 \}$, $\D_1 = \{ A(1), A(2), B(2), C(2) \}$, $\D_2 = \{ A(1), B(1), C(1) \}$, and $\sigma = \{A, B, C\}$, we have the following $\mathfrak O$:
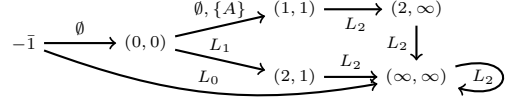
\centerline{
\begin{tikzpicture}[->,thick,node distance=1.8cm, transform shape,scale=0.9]\footnotesize

\node[rectangle] (s-1) {\scriptsize$-\bar{1}$};

\node[rectangle, right of = s-1] (s00) {\scriptsize$(0,0)$} edge[<-] node[above] {\scriptsize$\emptyset$} (s-1);

\node[rectangle] (s11) at (4,.5) {\scriptsize$(1,1)$} edge[<-] node[above] {\scriptsize$\emptyset, \{A\}$} (s00);

\node[rectangle] (s21) at (4,-.5) {\scriptsize$(2,1)$} edge[<-] node[above] {\scriptsize$L_1$} (s00);

\node[rectangle, right of = s11] (s2i)  {\scriptsize$(2, \infty)$}
edge[<-] node[below] {\scriptsize$\ L_2$} (s11)
;

\node[rectangle, right of = s21] (sii)  {\scriptsize$(\infty,\infty)$}
edge[<-] node[above] {\scriptsize$L_2$} (s21)
edge[<-] node[left] {\scriptsize$L_2$} (s2i)
edge[<-, out = 190, in = -20] node[above] {\scriptsize$L_0$} (s-1)
edge[<-, loop right] node[left] {\scriptsize$L_2$} (sii);
\end{tikzpicture}
}\\
with $L_0 \in 2^\sigma \setminus \{ \emptyset \}$, $L_1 \in \{ \{B\}, \{C\}, \{B,C\}, \{A, B\},$ $\{A,C\}, \Sigma\}$, and $L_2 = 2^\sigma$ (we omit nodes unreachable from $-\bar{1}$ such as $(\infty,0)$).
\end{example}

Given a path
\begin{equation}\label{eq:prod-path2}
\pi = -\bar{1} \to_{\rho_0} \avec{p}_0 \dots \to_{\rho_n}\avec{p}_n
\end{equation}
we call it a \emph{separating path} for $E$ if $\avec{p}_n = (\avec{n}, \bar \infty)$ and $\rho_n \ne \emptyset$. We denote by $\q_\pi$ the formula $\rho_0 \land \Diamond (\rho_1 \land \Diamond (\rho_2 \land \dots \land \Diamond \rho_n))$ which is not necessarily a query (since we assume the queries are in the normal form).
Now, it directly follows from the construction that
\begin{itemize}
\item a query $\q$ is a separator iff there exists a separating path $\pi$ in $\mathfrak O$ such that $\q = \q_\pi$
\end{itemize}
and the following holds (cf. Lemma~\ref{th:strongest-crit}):

\begin{lemma}\label{th:ugeneral-crit}
A $\Qpd$-query $\q$ is a unique most general $\Qpd$-separator for $E$ iff there is a separating path $\pi$ such that $\q = \q_\pi$ and $\q_\nu \models \q_\pi$, for every separating path $\nu$ with $\rho$.
\end{lemma}

We could then use a marking algorithm analogous to that for finding a strongest separator. The problem is that, unless $\sigma$ is fixed, $\mathfrak O$ has exponentially many edges w.r.t.\ $|\sigma|$. To provide an algorithm with the complexity $O(|\sigma| t_+^{c_+} t_-^{c_-})$, we define $\mathfrak O'$ with the same vertices as $\mathfrak O$ but with two sorts of labelled edges $\hookrightarrow_\rho$ and $\leadsto_\rho$ from $-\bar{1}$ to $\avec{p}$ and from $\avec{p}$ to $\avec{p}'$.
We set $\rho_{-\bar 1, \avec{p}} = \bigcap \{\rho \mid -\bar{1} \to_{\rho} \avec{p}\}$ and $\rho_{\avec{p}, \avec{p}'} = \bigcap \{\rho \mid \avec{p} \to_{\rho} \avec{p}'\}$.  Clearly for any $\rho' \subseteq \sigma$,
\begin{align}
 \label{eq:intersection1}\rho_{-\bar 1, \avec{p}} \models \rho' &\text{ iff }\rho \models \rho' \text{ for each }-\bar{1} \to_{\rho} \avec{p}\\
 \label{eq:intersection2}\rho_{\avec{p}, \avec{p}'} \models \rho' &\text{ iff }\rho \models \rho' \text{ for each }\avec{p} \to_{\rho} \avec{p}'
\end{align}

We let $-\bar{1} \hookrightarrow_{\rho} \avec{p}$ (respectively, $-\bar{1} \leadsto_\rho \avec{p}$) if $-\bar{1} \to_{\rho_{-\bar 1, \avec{p}}} \avec{p}$ (respectively, if $-\bar{1} \to \avec{p}$ and $-\bar{1} \not \to_{\rho_{-\bar 1, \avec{p}}} \avec{p}$) and $\avec{p} \hookrightarrow_{\rho} \avec{p}'$ (respectively, $\avec{p} \leadsto_\rho \avec{p}'$) if $\avec{p} \to_{\rho_{\avec{p}, \avec{p}'}} \avec{p}$ (respectively, if $\avec{p} \to \avec{p}'$ and $\avec{p} \not \to_{\rho_{\avec{p}, \avec{p}'}} \avec{p}'$). To illustrate, let $E = (E^+, E^-)$ with $E^+ = \{\D_1\}$ and $E^- = \{ \mathcal E_1, \mathcal E_2\}$ be as below

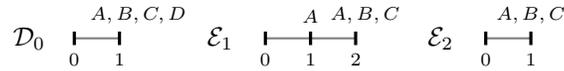
\centerline{
\begin{tikzpicture}[nd/.style={draw,thick,circle,inner sep=0pt,minimum size=1.5mm,fill=white},xscale=0.6]
\draw[thick,gray,-] (0,0) -- (1,0);

\node at (-1,0) {$\mathcal D_0$};

\slin{0,0}{}{\scriptsize$0$};
\slin{1,0}{\scriptsize$\quad \quad A,B,C, D$}{\scriptsize$1$};
\end{tikzpicture}
\begin{tikzpicture}[nd/.style={draw,thick,circle,inner sep=0pt,minimum size=1.5mm,fill=white},xscale=0.6]
\draw[thick,gray,-] (0,0) -- (2,0);

\node at (-1,0) {$\mathcal E_1$};
\slin{0,0}{}{\scriptsize$0$};
\slin{1,0}{\scriptsize$A$}{\scriptsize$1$};
\slin{2,0}{\quad\scriptsize$A, B, C$}{\scriptsize$2$};
\end{tikzpicture}
\
\begin{tikzpicture}[nd/.style={draw,thick,circle,inner sep=0pt,minimum size=1.5mm,fill=white},xscale=0.6]
\draw[thick,gray,-] (0,0) -- (1,0);

\node at (-1,0) {$\mathcal E_2$};

\slin{0,0}{}{\scriptsize$0$};
\slin{1,0}{\scriptsize$A,B,C$}{\scriptsize$1$};
\end{tikzpicture}
}
Then, we have $(0_0, 0_1, 0_2) \hookrightarrow_\emptyset (1_0, 1_1, 1_2)$ and $(0_0, 0_1, 0_2) \leadsto_{\emptyset} (1_0, 2_1, 1_2)$. We also observe that if $C$ was removed from $\mathcal E_1$ and $\mathcal E_2$, then it would be the case $(0_0, 0_1, 0_2) \hookrightarrow_{\{ B \}} (1_0, 2_1, 1_2)$.
A \emph{path} in $\mathfrak O'$ is a sequence the form~\eqref{eq:prod-path2} with $\Rightarrow_{\rho_i} \in \{ \hookrightarrow_{\rho_i}, \leadsto_{\rho_i}\}$ in place of $\to_{\rho_i}$. A $\mathfrak O'$-path $\pi$ is separating if $\avec{p}_n= (\avec{n}, \bar \infty)$ and $\avec{l}(\avec{p}_n) \ne \emptyset$.
For a path in $\mathfrak O'$ we define the formula $\q_\pi$ (which is not necessarily a query) as above, except if $\Rightarrow_{\rho_n} = \hookrightarrow_{\emptyset}$ and $|\avec{l}(\avec{p}_n)| =1$, we set $\rho_n = \{A \mid A \in \avec{l}(\avec{p}_n)\}$.
We now show:

\begin{lemma}\label{th:most-gen-poly}
$\q$ is a unique most general $\Qpd$-separator for $E$ iff $\q = \q_{\pi'}$ for a $\hookrightarrow$-separating path $\pi'$ in $\mathfrak O'$ such that $\Rightarrow_{\rho_n} = \hookrightarrow_{\emptyset}$ implies $|\avec{l}(\avec{p}_n)| =1$ and $\q_{\nu'} \models \q_{\pi'}$ for each separating path $\nu'$ in $\mathfrak O'$.
\end{lemma}
\begin{proof}
$(\Rightarrow)$ By Lemma~\ref{th:ugeneral-crit}, there exists a separating path $\pi$ of the form~\eqref{eq:prod-path2} in $\mathfrak O$ such that $\q = \q_\pi$ and $\q_\nu \models \q_\pi$ for each separating path $\q_\nu$ in $\mathfrak O$. Let the sequence of vertices of $\pi$ be $-\bar{1}, \avec{p}_0, \dots, \avec{p}_n$. We claim that
$$\pi' = -\bar{1} \hookrightarrow_{\rho_{-\bar{1}, \avec{p}_0}} \avec{p}_0 \hookrightarrow_{\rho_{\avec{p}_0, \avec{p}_1}} \avec{p}_1 \dots \hookrightarrow_{\rho_{\avec{p}_{n-1}, \avec{p}_{n}}} \avec{p}_{n}$$
is a separating path in $\mathfrak O'$. Indeed, since $-\bar{1} \to_{\rho_0} \avec{p}_0$  in $\mathfrak O$, it follows that $-\bar{1} \Rightarrow_{\rho_{-\bar{1}, \avec{p}_0}} \avec{p}_0$ for $\Rightarrow \in \{\hookrightarrow, \leadsto\}$. To see that $-\bar{1} \hookrightarrow_{\rho_{-\bar{1}, \avec{p}_0}} \avec{p}_0$, consider all the paths $\nu = -\bar{1} \to_{\rho} \avec{p}_0 \to_{\rho_1} \avec{p}_1 \dots \to_{\rho_n}\avec{p}_n$, where $\rho_i$ are as in $\pi$. Because they are all separating and $\q_\nu \models \q_\pi$, it follows that $\rho \models \rho_0$ for each $-\bar{1} \to_{\rho} \avec{p}_0$. It follows by~\eqref{eq:intersection1} that $\rho_{-\bar{1}, \avec{p}_0} \models \rho_0$. Obviously, $\rho_0 \models \rho_{-\bar{1}, \avec{p}_0}$, so $\rho_0 =\rho_{-\bar{1}, \avec{p}_0}$ and $-\bar{1} \hookrightarrow_{\rho_{-\bar{1}, \avec{p}_0}} \avec{p}_0$. Similarly, using the fact that $\avec{p}_{i-1} \to_{\rho_i} \avec{p}_i$ is in $\mathfrak O$ and~\eqref{eq:intersection2}, we obtain $\rho_i =\rho_{\avec{p}_{i-1}, \avec{p}_i}$ and $\avec{p}_{i-1} \hookrightarrow_{\rho_{\avec{p}_{i-1}, \avec{p}_i}} \avec{p}_i$ for $i < n$. To show the same for $i = n$, consider first the case $\avec{p}_{i-1} \not \to_{\emptyset} \avec{p}_i$. Then all the paths $\nu = -\bar{1} \to_{\rho_0} \avec{p}_0 \to_{\rho_1} \avec{p}_1 \dots \to_{\rho_{n-1}} \avec{p}_{n-1} \to_{\rho}\avec{p}_n$, where $\rho_i$ are as in $\pi$, are separating, so $\q_\nu \models \q_\pi$ and $\rho \models \rho_n$. Using the argument above we obtain $\rho_n =\rho_{\avec{p}_{n-1}, \avec{p}_n}$ and $\avec{p}_{n-1} \hookrightarrow_{\rho_{\avec{p}_{n-1}, \avec{p}_n}} \avec{p}_n$. Alternatively, if $\avec{p}_{i-1} \to_{\emptyset} \avec{p}_i$, we readily obtain $\avec{p}_{n-1} \hookrightarrow_{\emptyset} \avec{p}_n$ from the definition of $\hookrightarrow$.
It remains show that $\rho_{\avec{p}_{n-1}, \avec{p}_{n}} = \emptyset$ implies $|\avec{l}(\avec{p}_n)| =1$. Suppose $\rho_{\avec{p}_{n-1}, \avec{p}_{n}} = \emptyset$, then $\avec{p}_{n-1} \to_\emptyset \avec{p}_{n}$. It follows by the construction of $\mathfrak O$ that $\nabla_{\avec{p}_{n-1}, \avec{p}_n} = \emptyset$. Suppose, for the sake of contradiction, that $|\avec{l}(\avec{p}_n)| \geq 2$ and distinct $A, B \in \avec{l}(\avec{p}_n)$. It follows that $\avec{p}_{n-1} \to_{\{A\}} \avec{p}_{n}$ and $\avec{p}_{n-1} \to_{\{B\}} \avec{p}_{n}$. We observe that both $\pi_1 = -\bar{1} \to_{\rho_0} \avec{p}_0 \dots \to_{\rho_{n-1}}\avec{p}_{n-1} \to_{\{A\}} \avec{p}_{n}$ and $\pi_2 = -\bar{1} \to_{\rho_0} \avec{p}_0 \dots \to_{\rho_{n-1}}\avec{p}_{n-1} \to_{\{B\}} \avec{p}_{n}$ are separating path in $\mathfrak O$. But because $\rho_n \neq \emptyset$ in $\pi$, it follows that either $\q_{\pi_1} \not \models \q_\pi$ or $\q_{\pi_2} \not \models \q_\pi$, which is a contradiction.

Now we will show that any separating path
$$\nu' = -\bar{1} \Rightarrow_{\rho_{-\bar{1}, \avec{p}_0'}} \avec{p}_0' \Rightarrow_{\rho_{\avec{p}_0', \avec{p}_1'}} \avec{p}_1' \dots \Rightarrow_{\rho_{\avec{p}_{l-1}', \avec{p}_{l}'}} \avec{p}_{l}'$$
in $\mathfrak O'$ is such that $\q_{\nu'} \models \q_{\pi'}$. Recall that the path $\pi$ with $\rho_0 = \rho_{-\bar{1}, \avec{p}_0}$, $\rho_i = \rho_{\avec{p}_{i-1}, \avec{p}_i}$ for $1 \leq i \leq n-1$, and $\rho_n = \rho_{\avec{p}_{n-1}, \avec{p}_n}$ if $\rho_{\avec{p}_{n-1}, \avec{p}_n} \neq \emptyset$ and $\rho_n = \avec{l}(\avec{p}_n)$ otherwise, is such that $\q_\nu \models \q_\pi$ for each separating path
$$\nu = -\bar{1} \to_{\rho_0^\nu} \avec{p}_0' \dots \to_{\rho_{l-1}^\nu} \avec{p}_{l-1}' \to_{\rho_{l}^\nu} \avec{p}_{l}'$$
in $\mathfrak O$. We need to following result:
\begin{lemma} There exists a sequence $0 = i_0 < i_1 < \dots < i_n \leq l$ such that $\rho_{i_0}^\nu \models \rho_0$, \dots, $\rho_{i_n}^\nu \models \rho_n$ for all $\nu$ in $\mathfrak O$ as above.
\end{lemma}
\begin{proof}
For $i_0 = 0$, the result is straightforward given that $\q_\nu \models \q_\pi$ for every $\nu$.
It will be convenient, for a path $\mu$ of the form $\avec{s}_0 \to_{\rho'_1} \dots \to_{\rho'_k} \avec{s}_k$ in $\mathfrak O$ (with $\avec{s}_0 \neq -\bar 1$), to define the formula $\q_\mu = \Diamond (\rho_1' \land \Diamond (\rho_2' \land \dots \land \Diamond \rho_k'))$ and for a path $\pi$ of the form~\eqref{eq:prod-path2} and $0 < i \leq n$, to define $\pi^i$ as $\avec{p}_{i-1} \to_{\rho_{i}} \dots \to_{\rho_n} \avec{p}_n$. We first establish that:
\begin{description}
  \item[$(i_1)$] there exists $i_1$ such that $\rho_{i_1}^\nu \models \rho_1$ for every $\nu$ and $\q_{\nu^{i_1}} \models \q_{\pi^1}$ for every $\nu$
\end{description}
Indeed, suppose for the sake of contradiction that for all $0 < i \leq l$ either $\rho_i^\nu \not \models \rho_1$ for some $\nu$ or $\q_{\nu^i} \not \models \q_{\pi^1}$ for some $\nu$. Take $i = 1$; if $\q_{\nu^1} \not \models \q_{\pi^1}$ for some $\nu$ then we immediately obtain a contradiction to $\q_{\nu} \models \q_{\pi}$ for all $\nu$. It means that $\rho_1^\nu \not \models \rho_1$ for some $\nu$. Denote $\rho_1^\nu$ by  $\rho_1^*$. Take $i = 2$; if $\q_{\nu^2} \not \models \q_{\pi^1}$ for some $\nu$ then we have a path $\avec{p}_0' \to_{\rho_1^*} \nu^2$ in $\mathfrak O$ such that $\q_{\avec{p}_0' \to_{\rho_1^*} \nu^2} \not \models \q_{\pi^1}$. This is a contradiction. It implies that $\rho_2^\nu \not \models \rho_1$ for some $\nu$. Denote $\rho_2^\nu$ by $\rho_2^*$. We proceed in this manner and suppose we have not obtained a contradiction until $i = l-1$. This implies that we have $\rho_1^*, \dots, \rho_{l-1}^*$ such that $\rho_1^* \not \models \rho_1, \dots, \rho_{l-1}^* \not \models \rho_1$. Take $i = l$; if $\q_{\nu^l} \not \models \q_{\pi^1}$ for some $\nu$ then we have a path $\avec{p}_0' \to_{\rho_1^*} \dots \to_{\rho_{l-1}^*} \nu^l$ in $\mathfrak O$ such that $\q_{\avec{p}_0' \to_{\rho_1^*} \dots \to_{\rho_{l-1}^*} \nu^l} \not \models \q_{\pi^1}$. This is a contradiction. It implies that $\rho_l^\nu \not \models \rho_1$ for some $\nu$ and we denote $\rho_l^\nu$ by $\rho_l^*$. Now we have a path $\avec{p}_0' \to_{\rho_1^*} \dots \to_{\rho_l^*} \avec{p}_l'$ in $\mathfrak O$ such that $\q_{\avec{p}_0' \to_{\rho_1^*} \dots \to_{\rho_l^*}  \avec{p}_l'} \not \models \q_{\pi^1}$, which is a contradiction. We established $(i_1)$. Now we can show:
\begin{description}
  \item[$(i_2)$] there exists $i_2 > i_1$ such that $\rho_{i_2}^\nu \models \rho_2$ for every $\nu$ and $\q_{\nu^{i_2}} \models \q_{\pi^2}$ for every $\nu$
\end{description}
Indeed, suppose for the sake of contradiction that for all $i_1 < i \leq l$ either $\rho_i^\nu \not \models \rho_2$ for some $\nu$ or $\q_{\nu^i} \not \models \q_{\pi^2}$ for some $\nu$. Take $i = i_1+1$; if $\q_{\nu^{i}} \not \models \q_{\pi^2}$ for some $\nu$ then we immediately obtain a contradiction to $\q_{\nu^{i_1}} \models \q_{\pi^1}$ for all $\nu$ established by $(i_1)$. It means that  $\rho_{i_1+1}^\nu \not \models \rho_2$ for some $\nu$. Denote $\rho_{i_1+1}^\nu$ by $\rho_{i_1+1}^*$. Similarly to the proof above, we continue exploring possible values of $i \leq l$. We will obtain a contradiction at latest for $i = l$ and the case $\rho_i^\nu \not \models \rho_2$ for some $\nu$. In that case, we will have $\q_{\avec{p}_{i_1}' \to_{\rho_{i_1+1}^*} \dots \to_{\rho_l^*}  \avec{p}_l'} \not \models \q_{\pi^2}$ which contradicts to $\q_{\nu^{i_1}} \models \q_{\pi^1}$ for all $\nu$.

Relying on $(i_2)$, we can prove $(i_3)$ with the same argument as above. After we established $(i_n)$, the proof of the lemma is complete.
\end{proof}

We now use the sequence $0 = i_0 < i_1 < \dots < i_n \leq l$ from the lemma above. It follows by~\eqref{eq:intersection1} and~\eqref{eq:intersection2} that $\rho_{-\bar{1}, \avec{p}_{i_0}} \models \rho_{0}$, $\rho_{\avec{p}_{i_1-1}', \avec{p}_{i_1}'} \models \rho_1$, \dots, $\rho_{\avec{p}_{i_{n-1}-1}', \avec{p}_{i_{n-1}}'} \models \rho_{n-1}$. If $i_n < l$ or $\rho_{\avec{p}_{l-1}', \avec{p}_{l}'} \neq \emptyset$, then it also follows $\rho_{\avec{p}_{i_{n}-1}', \avec{p}_{i_{n}}'} \models \rho_{n}$. Suppose $i_n = l$ and $\rho_{\avec{p}_{l-1}', \avec{p}_{l}'} = \emptyset$. It follows that, $\avec{l}(\avec{p}_{l}') = \{A\}$, for some $A \in \sigma$ and $A \models \rho_n$ (otherwise, $\avec{p}_{l-1}'\to_{\{A\}} \avec{p}_{l}'$ and $\avec{p}_{l-1}'\to_{\{B\}} \avec{p}_{l}'$ for distinct $A, B \in \sigma$ and then $i_n < l$). In either case, we can conclude that $\q_{\nu'} \models \q_{\pi}'$ and the proof of $(\Rightarrow)$ is done.

$(\Leftarrow)$ Let $\pi'$ above be such that $\q_{\nu'} \models \q_{\pi'}$ for each separating path $\nu'$ in $\mathfrak O'$. Consider the path $\pi$ in $\mathfrak O$ of the form~\eqref{eq:prod-path2} with $\rho_0 =\rho_{-\bar{1}, \avec{p}_0}$ and $\rho_i =\rho_{\avec{p}_{i-1}, \avec{p}_i}$ for $1 \leq i \leq n-1$, and $\rho_n = \rho_{\avec{p}_{n-1}, \avec{p}_n}$ if $\rho_{\avec{p}_{n-1}, \avec{p}_n} \neq \emptyset$ and $\rho_n = \avec{l}(\avec{p}_n)$ otherwise. Clearly, by the definition of $\hookrightarrow$, this path is separating. Consider a separating path $\nu$ in $\mathfrak O$ in the form above; we will show that $\q_\nu \models \q_\pi$. Consider the path $\nu'$ in $\mathfrak O'$ as above; clearly, it is separating (in $\mathfrak O'$) and so it follows that $\q_{\nu'} \models \q_{\pi'}$. Thus, there exists a sequence $0 = i_0 < i_1 < \dots < i_n \leq l$ such that $\rho_{\avec{p}_{i_0-1}', \avec{p}_{i_0}'} \models \rho_0$, \dots, $\rho_{\avec{p}_{i_{n-1}-1}', \avec{p}_{i_{n-1}}'} \models \rho_{n-1}$, and either $\rho_{\avec{p}_{i_n-1}', \avec{p}_{i_n}'} \models \rho_n$ or $i_n = l$, $\rho_{\avec{p}_{l-1}', \avec{p}_{l}'} = \emptyset$, $\avec{l}(\avec{p}_{l}') = \{A\}$, and $A \models \rho_n$. In either case, from~\eqref{eq:intersection1} and~\eqref{eq:intersection2} we obtain that $\rho_{i_0}^\nu \models \rho_0$, \dots, $\rho_{i_n}^\nu \models \rho_n$, which completes the proof.
\end{proof}

Now we explain how, for a given $\avec{p}, \avec{p}'$ in $\mathfrak O$ (or $\mathfrak O$), to compute $\rho_{-\bar 1, \avec{p}}$ and $\rho_{\avec{p}, \avec{p}'}$ in time polynomial in $|\sigma| |\mathfrak O|$. Observe that $\avec{p} \to_\rho \avec{p}'$ iff $\rho$ is a (propositional) separator for $(\{ \avec{l}(\avec{p}')\}, \{ \avec{l}(\avec{p}'') \mid \avec{p}'' \in \nabla_{\avec{p}, \avec{p}'}\})$ (and analogously for $-\bar{1} \to_\rho \avec{p}$). Then, it can be readily checked that:
\begin{equation*}
A \in \rho_{\avec{p}, \avec{p}'} \text{ iff there is }\avec{p}'' \in \nabla_{\avec{p}, \avec{p}'} \text{ s.t. }
\{A\} = \sepmg(( \{\avec{l}(\avec{p}')\}, \{\avec{l}(\avec{p}'')\}),\mathcal{Q})
\end{equation*}
for each $A \in \sigma$, where $\mathcal{Q}$ is the class of the propositional queries. Thus, $\rho_{\avec{p}, \avec{p}'}$ can be computed in \PTime{}. Moreover, we have $\avec{p} \hookrightarrow_{\rho_{\avec{p}, \avec{p}'}} \avec{p}'$ if $\rho_{\avec{p}, \avec{p}'}$ separates $(\{ \avec{l}(\avec{p}')\}, \{ \avec{l}(\avec{p}'') \mid \avec{p}'' \in \nabla_{\avec{p}, \avec{p}'}\})$ and otherwise we have $\avec{p} \leadsto_{\rho_{\avec{p}, \avec{p}'}} \avec{p}'$.

Checking the condition of Lemma~\ref{th:most-gen-poly} on $\mathfrak O'$ is done similarly to the algorithm used in Lemma~\ref{th:strongest-crit}, marking each $\avec{p}$ in $\mathfrak O'$ by either a query $\q_\pi$ for some $\hookrightarrow$-path $\pi$ (from $-\bar{1}$) or with $\varepsilon$. Intuitively, $\q_\pi$ marks $\avec{p}$ if $\q_\nu \models \q_\pi$. The details, which are very similar to the proof of in Lemma~\ref{th:strongest-crit}, are left to the reader.

\medskip

We next prove Theorem~\ref{thm:algorithms} $(a)$ for $\Qpnd$-separators.
Let $E^+ = \{\D_1, \dots, \D_{c_+} \}$. As in the previous part of the proof, we denote by $\avec{n}$ any vector $(n_1, \dots, n_{c^+})$ with $0 \leq n_j \leq \max \D_j$, $1 \leq j \leq c_+$. The vertices of $\mathfrak P$ are tuples $\avec{n}^{t}$, where $0 \leq t \leq \min \{ \max \D_j - n_j \mid 1 \leq j \leq c_+\}$. Define $\avec{n}+t = (n_1+t, \dots, n_c+t)$ and $\avec{n} < \avec{n}'$ if $n_j < n_j'$ for all $1 \leq j \leq c$. $\mathfrak P$ has edges $\avec{n}^t \to \avec{n}_1^s$ for all vertices $\avec{n}^t, \avec{n}_1^s$ in $\mathfrak P$ such that $\avec{n}+t < \avec{n}_1$. 
We set the label $\avec{l}(\avec{n}^t)$ of each $\avec{n}^t$ in $\mathfrak P$ to be a $t+1$-component vector $(\avec{l}(\avec{n}), \avec{l}(\avec{n}+1), \dots, \avec{l}(\avec{n}+t))$ (recall that $\avec{l}(\avec{n}) = \{A \in \sigma \mid A(n_j) \in \D_j \text{ for all $j$}\}$).
Similarly, given $E^- = \{\mathcal D_1, \dots, \mathcal D_{c_-} \}$, we denote by $\avec{m}$ any vector $(m_1, \dots, m_{c_-})$ with $m_j \in [0, \max \mathcal D_j] \cup \{ \infty \}$, $1 \leq j \leq c_-$. Let $m_j+t = \infty$ if $m_j + t > \max \mathcal D_j$ and $\infty + t = \infty$ for $t \geq 0$. Recall that we denote $(\infty, \dots, \infty)$ by $\bar{\infty}$. The vertices of $\mathfrak N$ are tuples of the form $\avec{m}^{t}$, where $0 \leq t \leq \min \{ \max \D \mid \D \in E^+\}$.
$\mathfrak N$ has edges $\avec{m}^t \to \avec{m}'^s$ for all vertices such that $\avec{m} + t < \avec{m}'$.
We define the label $\avec{l}(\avec{m}^t)$ of each $\avec{m}^t$ in $\mathfrak N$ to be a $t+1$-component vector $(\avec{l}(\avec{m}), \avec{l}(\avec{m}+1), \dots, \avec{l}(\avec{m}+t))$; recall that $\avec{l}(\avec{m}) = \{A \in \sigma \mid A(m_j) \in \D_j \text{ for all $j$ such that } m_j \neq \infty\}$.

We take a product-like structure $\mathfrak P \otimes \mathfrak N$ containing the vertices $(\bar{0}^t, \avec{m}^t)$ such that $m_j = 0$ if $\{A(0) \mid A \in \avec{l}(\bar 0)\} \cup \dots \cup \{A(t) \mid A \in \avec{l}(\bar t)\} \subseteq \D_j$ and $m_j = \infty$ otherwise, for all $1 \leq j \leq c_-$. Also, it contains vertices $(\avec{n}^t, \avec{m}^t)$ with $\avec{n}^t \neq \bar{0}^t$ from $\mathfrak P$ and with $\avec{m}^t$ from $\mathfrak N$ such that $\avec{l}(\avec{n}^t) \subseteq \avec{l}(\avec{m}^t)$.
The edges of $\mathfrak P \times \mathfrak N$ are $(\avec{n}^t, \avec{m}^t) \to (\avec{n}'^s, \avec{m}'^s)$ such that $\avec{n}^t \to \avec{n}'^s$, $\avec{m}^t \to \avec{m}'^s$ and there exists no $\mathfrak P \otimes \mathfrak N$-vertex $(\avec{n}'^s, \avec{m}''^s)$ satisfying $\avec{m}^t \to \avec{m}''^s$, $\avec{m}_i'' < \avec{m}_i'$ and $\avec{m}_i'' \neq \infty$ for some $1 \leq i \leq l$. The latter \emph{minimality} condition on $\avec{m}'$ ensures that for given $(\avec{n}^t, \avec{m}^t)$ and $\avec{n}'^s \leftarrow \avec{n}^t$ there is a unique $\avec{m}'^s$ such that $(\avec{n}^t, \avec{m}^t) \to (\avec{n}'^s, \avec{m}'^s)$.

Given a path in $\mathfrak P \otimes \mathfrak N$
\begin{equation}\label{eq:prod-path-nxt}
\pi = (\avec{n}_0^{t_0}, \avec{m}_0^{t_0}), \dots, (\avec{n}_\ell^{t_\ell}, \avec{m}_\ell^{t_\ell})
\end{equation}
with $\avec{n}_0 = \bar{0}$, let $\q_\pi$ be the $\Qpnd$-query of the form~\eqref{dnpath} with $n = t_0 + \dots + t_\ell + \ell$ such that:
\begin{itemize}
  \item $(\rho_0, \dots, \rho_{t_0}) = \avec{l}(\avec{n}_0^{t_0})$ and $\op_i = \nxt$ for $i \in (0, t_0]$
  \item $(\rho_{t_0+1}, \dots, \rho_{t_0+t_1+1}) = \avec{l}(\avec{n}_1^{t_1})$, $\op_{t_0+1} = \Diamond$ and $\op_i = \nxt$ for $i \in (t_{0}+1, t_{0}+t_1+1]$
  \item $\dots$
  \item $(\rho_{t_0+ \dots +t_{\ell-1}+\ell}, \dots, \rho_{t_0+ \dots +t_{\ell}+\ell}) = \avec{l}(\avec{n}_\ell^{t_\ell})$, $\op_{t_0+ \dots +t_{\ell-1}+\ell} = \Diamond$ and $\op_i = \nxt$ for $i \in (t_0+ \dots +t_{\ell-1}+\ell, t_0+ \dots +t_{\ell}+\ell]$.
\end{itemize}
(Note that $\q_\pi$ is not necessarily in normal form.) If $\avec{l}(\avec{n}_\ell+t_\ell) \neq \emptyset$ and $\avec{m}_\ell+t_\ell = \bar{\infty}$, we call $\pi$ a \emph{separating} path.

\begin{lemma}\label{th:reach-with-nxt}
$E$ is $\Qpnd$-separable iff there exists a separating path in $\mathfrak P \otimes \mathfrak N$.
\end{lemma}

\begin{proof}
$(\Rightarrow)$ Suppose $\q = \bar \rho_0 \land \Diamond (\bar \rho_1 \land ( \dots \Diamond \bar \rho_\ell) \dots )$ (in the normal form) separates $E$, where each $\bar \rho_i = \rho_i^0 \land \nxt \rho_i^1 \land \dots \land \nxt^{p_i} \rho_i^{p_i}$. If $\ell = 0$, we can further assume $\rho_0^{p_0} \neq \emptyset$. It follows
$\bar \rho_0 \subseteq \avec{l}(\bar{0}^{p_0})$ and $\bar \rho_0 \not \subseteq (t_{\mathcal E_j}(0), \dots, t_{\mathcal E_j}(p_0))$ for all $1 \leq j \leq l$.  Thus, $\avec{l}(\bar{0}+p_0) \neq \emptyset$, and $(\bar{0}^{p_0}, \bar{\infty}^{p_0})$ (so, $\infty + p_0 = \bar{\infty}$) is in $\mathfrak P \otimes \mathfrak N$, which completes the proof.

Suppose $\ell > 0$, $\bar \rho_0 \land \Diamond (\bar \rho_1 \land ( \dots \Diamond \bar \rho_{\ell-1}) \dots )$ does not separate $E$, and $\rho_\ell^{p_\ell} \neq \emptyset$.
We define $\avec{n}_0 = \bar{0}$. Assuming that $\avec{n}_i = (n_1, \dots, n_k)$ has been defined, define $\avec{n}_{i+1} = (n_1', \dots, n_k')$ where $n_j' = \min \{ n \mid n > n_j+p_i, \bar \rho_{i+1} \subseteq (t_{\mathcal D_j}(n), t_{\mathcal D_j}(n+1), \dots, t_{\mathcal D_j}(n+p_{i+1})) \}$, for $0 \leq i < \ell$. Clearly, we have $\avec{n_0}^{p_0} \to \dots \to \avec{n_\ell}^{p_\ell}$ and $\avec{l}(\avec{n_\ell}+p_\ell) \neq \emptyset$.
Then, we take $\avec{m}_0$ equal to $\avec{m}$ satisfying $m_j = 0$ if $(t_{\D_j}(0), \dots, t_{\D_j}(p_0)) \subseteq (t_{\mathcal E_j}(0), \dots, t_{\mathcal E_j}(p_0))$, and $m_j = \infty$ otherwise, for all $1 \leq j \leq l$. We set $\avec{m}_{i+1}$ to be $\avec{m}$ such that $(\avec{n}_i^{p_i}, \avec{m}_i^{p_i}) \to (\avec{n}_{i+1}^{p_{i+1}}, \avec{m}^{p_{i+1}})$. We claim that the path $(\avec{n}^{p_0}_0, \avec{m}^{p_0}_0), \dots, (\avec{n}^{p_\ell}_\ell, \avec{m}^{p_\ell}_\ell)$ is such that $\avec{m}_\ell + {p_\ell} = \bar{\infty}$. To prove the claim, set $d_{i,j} = \min \{ t \mid t \geq i + p_0 + \dots + p_{i} \text{ and } \mathcal E_j^{\leq t} \models \bar \rho_0 \land \Diamond (\bar \rho_1 \land ( \dots \Diamond \bar \rho_i) \dots )\}$ or $d_{i,j} = \infty$ if $t$ as above does not exist, for $0 \leq i \leq \ell$ and $1 \leq j \leq l$. Let $\avec{d}_i = (d_{i,1}, \dots, d_{i,l})$. It follows that $\avec{d}_\ell= \bar{\infty}$. It is easily shown by induction on $i$ that $\avec{d}_i = \avec{m}_i+p_i$ which completes the proof.

$(\Leftarrow)$ Let \eqref{eq:prod-path-nxt} be a separating path.
We will show that $\q_\pi$ separates $E$. Indeed, clearly $\D \models \q_\pi$ for each $\D \in E^+$. By induction on $i$, we show that $\avec{d}_i = \avec{m}_i+t_i$. It follows that $\mathcal E \not \models \q_\pi$ for all $\mathcal E \in E^-$ which completes the proof.
\end{proof} %

The existence of a separating path in $\mathfrak P \otimes \mathfrak N$ can be checked in time $O(\tpcp \tmcm)$. If $\cp$ and $\cm$ are bounded, given $(\avec{n}^t, \avec{m}^t)$ and $(\avec{n}'^s, \avec{m}'^s)$, we can check in logspace whether $(\avec{n}^t, \avec{m}^t)$ is the root of $\mathfrak P \times \mathfrak N$ and $(\avec{n}^t, \avec{m}^t) \to (\avec{n}'^s, \avec{m}'^s)$. So the existence of a separating path can be decided in \NL.

\medskip

We now prove Theorem~\ref{thm:algorithms} $(b)$ for $\Qpnd$-separators.
As it was already done in the proof of Lemma~\ref{th:reach-with-nxt}, we assume that the $\Qpnd$-queries are
\begin{equation}\label{eq:path-vect-form}
  \q = \bar \rho_0 \land \Diamond (\bar \rho_1 \land  \dots  \land \Diamond \bar \rho_\ell),
\end{equation}
where each $\bar \rho_i = \rho_i^0 \land \nxt \rho_i^1 \land \dots \land \nxt^{p_i} \rho_i^{p_i}$. Furthermore, we treat each query $\bar \rho_i$ as a vector $(\rho_i^0, \dots, \rho_i^{p_i})$ when convenient. The proof of the following statement is a straightforward modification of the proof of Lemma~\ref{th:bounded-subs} (in the same way as the proof of Lemma~\ref{th:reach-with-nxt} modifies the proof of Lemma~\ref{lem:sep-bounded-diamond-prod}):
\begin{lemma}\label{th:bounded-subs}
If $\q$ of the form~\eqref{eq:path-vect-form} separates $E$, then there exists a separating path~\eqref{eq:prod-path-nxt} such that $t_i = p_i$ and
$\bar \rho_i \subseteq \avec{l}(\avec{n}_i^{t_i})$, for all $i \le \ell$.
\end{lemma}

In essence, we can now use any existing algorithm for computing a shortest/longest path in the directed acyclic graph\footnote{see e.g. Sedgewick, Robert; Wayne, Kevin Daniel (2011), Algorithms (4th ed.), Addison-Wesley Professional.} $\mathfrak P \times \mathfrak N$, where the edges $(\avec{n}^t, \avec{m}^t) \to (\avec{n}'^s, \avec{m}'^s)$ have weights $s+1$ (we also need to add a new root and extra edges leading to the root $(\bar{0}^s, \avec{m}^s)$ of $\mathfrak P \otimes \mathfrak N$). Some straightforward algorithms are provided below.

For the shortest separator, we compute the closures $\mathcal C^\ell$ containing all the tuples $(\avec{n}'^t, \avec{m}'^t, w)$ such that $(\avec{n}'^t, \avec{m}'^t)$ is reachable in $\mathfrak P \times \mathfrak N$ from some $(\bar{0}^s, \avec{m}^s)$ on a path of length $\leq \ell$, while
\begin{equation}\label{eq:closure-weight}
w = \min \{ \ell' + t_1 + \dots + t_{\ell'} \mid (\bar{0}^{t_1}, \avec{m}^{t_1}), \dots, (\avec{n}'^{t_{\ell'}}, \avec{m}'^{t_{\ell'}})
\text{ is a path from some } (\bar{0}^s, \avec{m}^s) \text{ to }(\avec{n}'^t, \avec{m}'^t) \}.
\end{equation}
We compute all $\mathcal C^1, \dots, \mathcal C^\ell$, where $\ell$ is the length of the longest path in $\mathfrak P \otimes \mathfrak N$. Let $\mathcal C^i_\infty$ be the restriction of $\mathcal C^i$ to tuples $(\avec{n}'^t, \avec{m}'^t, w)$ with $\avec{m}'+t = \bar \infty$ and $\avec{l}(\avec{n}'+t) \neq \emptyset$.  We select $i^*$ and $w^*$ such that $(\avec{n}'^t, \avec{m}'^t, w^*) \in \mathcal C^{i^*}_\infty$ and $w^* = \min \{ w \mid (\avec{n}'^t, \avec{m}'^t, w) \in \mathcal C^i_\infty, \ 1 \leq i \leq \ell \}$. In essence, $w^*$ is the length of the shortest query separating $E$. To construct a $\q$ with length $w^*$, we use the sets $\mathcal C^1, \dots, \mathcal C^{i^*-1}, \mathcal C^{i^*}_\infty$. $\varkappa$ will be $\bar \rho_1 \land \Diamond (\bar \rho_2 \land ( \dots \Diamond \bar \rho_{i}) \dots )$ where each $\bar \rho_i = \rho_i^0 \land \nxt \rho_i^1 \land \dots \land \nxt^{p_i} \rho_i^{p_i}$. Take an $\avec{n}^{t_{i^*}}_{i^*}$ such that $(\avec{n}^{t_{i^*}}_{i^*}, \avec{m}'^{t_{i^*}}, w^*) \in \mathcal C^{i^*}_\infty$. We set $p_{i^*} = t_{i^*}$ and $(\rho_{i^*}^0, \dots, \rho_{i^*}^{p_{i^*}}) = \avec{l}(\avec{n}^{t_{i^*}}_{i^*})$. Now, take an $\avec{n}^{t_{i^*-1}}_{i^*-1}$ such that $(\avec{n}^{t_{i^*-1}}_{i^*-1}, \avec{m}^{t_{i^*-1}}, w_{i^*-1}) \in \mathcal C^{i^*-1}$, $(\avec{n}^{t_{i^*-1}}_{i^*-1}, \avec{m}^{t_{i^*-1}}) \to (\avec{n}^{t_{i^*}}_{i^*}, \avec{m}'^{t_{i^*}})$, and $w^* = w_{i^*-1}+t_{i^*}+1$. We set $p_{i^*-1} = t_{i^*-1}$ and $(\rho_{i^*-1}^0, \dots, \rho_{i^*-1}^{p_{i^*-1}}) = \avec{l}(\avec{n}^{t_{i^*-1}}_{i^*-1})$. We continue in this fashion to obtain $t_1$ such that $(\bar{0}^{t_1}, \avec{m}^{t_1}, t_1+1) \in \mathcal C^{1}$ and $w_2 = t_1+t_{2}+2$. We define $p_1 = t_1$ and $(\rho_{1}^0, \dots, \rho_{1}^{p_{1}}) = \avec{l}(\bar{0}^{t_{1}})$. This completes the definition of a shortest $\q$ separating $E$.

To compute a longest separator, we compute the closures $\mathcal C^i$ as in the case of the shortest separator except for using $\max$ instead of $\min$ in~\eqref{eq:closure-weight}. We proceed with computing the separator $\q$ as in the case above setting $w^*$ with $\max$ instead of $\min$.

\medskip

Finally, we prove Theorem~\ref{thm:algorithms} $(c)$ for unique most specific $\Qpnd$-separators.
First, we observe that the proof of Lemma~\ref{th:strongest-crit} is modified in a straightforward way to demonstrate that this lemma holds for $\Qpnd$-queries and separating paths~\eqref{eq:prod-path-nxt} in $\mathfrak P \otimes \mathfrak N$.
Secondly, we note that checking the existence of a separating path $\pi$ such that $\q_\pi \models \q_\nu$ for every separating path $\nu$ in $\mathfrak P \otimes \mathfrak N$ can be checked in \PTime{} by a very similar algorithm as for $\Qpd$-queries. To see that the marking of the nodes can still be done in \PTime{} for $\Qpnd$, we only need to note that checking $\q \models \q'$ for $\q, \q' \in \Qpnd$ remains in \PTime{} (see Theorem~\ref{thm:containment}). Moreover, the problem of checking whether there is a containment witness for a given $\q_{\pi'_1}$ and $\q_\pi$ for every $\pi$ in the set $\Pi$ of path~\eqref{eq:prod-path-nxt} in $\mathfrak P \otimes \mathfrak N$ ending at a given $(\avec{n}^t, \avec{m}^t)$, is in $\NL$.

At the final stage of the algorithm, we need to consider all the nodes of the form $(\avec{n}_i^{t_i}, \avec{m}_i^{t_i})$, $i = 1,\dots,l$, with $\avec{l}(\avec{n}_i+t_i) \neq \emptyset$, $\avec{m}_i+t_i = \bar{\infty}$ and marked with $\q_i \neq \varepsilon$. If there are no such, there is no unique most specific separator. Otherwise, take the set $\Pi$ of all paths leading to the $(\avec{n}_i^{t_i}, \avec{m}_i^{t_i})$. We check whether there is a containment witness for $\q_1$ and every $\q_\pi$, $\pi \in \Pi$, in which case $\q_1$ is returned as a unique most specific separator. Otherwise, we do the same for $\q_{2}$, and so on.

\end{proof}

\section{Proofs for Section~\ref{sec:conc}}
We show that for $\Q\in\{\Qpd,\Qpnd,\Qint\}$, verifying longest and shortest separators is \coNP-complete. The upper bound is trivial, and so we focus on the lower bound. We start with a general construction that shows how one can combine two example sets in a controlled way.

Let $\Q\in\{\Qpd,\Qpnd,\Qint\}$ be given and consider example set $E_{1}=(E_{1}^{+},E_{1}^{-})$ and $E_{2}=(E_{2}^{+},E_{2}^{-})$ with disjoint signatures such that all separating queries in $\sep(E_{1},\Q)$ and $\sep(E_{2},\Q)$ start with $\Diamond$.
Define $E^{+}$ by taking all
	$$
	\D_{1}\D_{2}, \quad \D_{2}\D_{1}
	$$
with $\D_{1}\in E_{1}^{+}$ and $\D_{2}\in E_{2}^{+}$
and let $E^{-} = E_{1}^{-} \cup E_{2}^{-}$.
\begin{lemma}
	$\sep(E,\Q)= \sep(E_{1},\Q) \cup \sep(E_{2},\Q)$.
\end{lemma}
\begin{proof}
	$\sep(E,\Q)\supseteq \sep(E_{1},\Q) \cup \sep(E_{2},\Q)$ follows from the condition
	that queries in $\sep(E_{1},\Q)$ and $\sep(E_{2},\Q)$ start with $\Diamond$.
	
	To show that $\sep(E,\Q)\subseteq \sep(E_{1},\Q) \cup \sep(E_{2},\Q)$ observe that it follows from the definition of $E^{+}$ that every
	query in $\sep(E,\Q)$ either only uses symbols from $E_{1}$ or only symbols from $E_{2}$. Now the inclusion is trivial.
\end{proof}
\begin{theorem}
	 Let $\Q\in\{\Qpd,\Qpnd,\Qint\}$. Then verification of longest and shortest separating queries is \coNP-hard.
\end{theorem}
\begin{proof}
	By reduction from SAT for the complement. Assume a propositional formula $\varphi$ is given. Construct $E_{1}$ as in the proof of Point~2 of Theorem~\ref{lem:consnew} such that there is a separating query in $\Q$ if $\varphi$ is satisfiable. Assume all separating queries have length at least 2 and length at most $n$.
	
	First let $E_{2}$ be any example set in a signature disjoint from the signature of $E_{1}$ and such that all separators have the form $\Diamond \rho$ with $\rho$ a set of atoms. Assume some such $\q$ separates $E_{2}$. Then $\q$ is not a longest separator of $E$ (as defined above) iff $\varphi$ is satisfiable.
	
	Now let $E_{2}$ be any example set in a signature disjoint from the signature of $E_{1}$ and such that all separators start with $\Diamond$ and have length $n+1$. Assume some such $\q$ separates $E_{2}$. Then $\q$ is not a shortest separator of $E$ (as defined above) iff $\varphi$ is satisfiable.
\end{proof}

\end{appendix}	
		
\end{document}